\documentclass[aop,MSNbibl,seceqn,citesort,dvips]{arximspdf}
\usepackage{mathbh}
\usepackage{graphicx}

\doi{10.1214/11-AOP726} 
\volume{41}
\issue{4}
\pubyear{2013}
\firstpage{2599}
\lastpage{2647}

\makeatletter

\renewcommand{\widetilde}{\tilde}

\newcommand{\cal}{\mathcal}

\newcommand{\lcont}{\swarrow\atop\searrow}
\newcommand{\rcont}{\nwarrow\atop\nearrow}

\newcommand{\Ai}{\operatorname{Ai}}
\newcommand{\Id}{\mathbh{1}}
\newcommand{\Or}{\mathcal{O}}
\newcommand{\Pb}{\mathbb{P}}
\newcommand{\I}{{\mathrm i}}
\newcommand{\C}{\mathbb{C}}
\newcommand{\R}{\mathbb{R}}
\newcommand{\N}{\mathbb{N}}
\newcommand{\Z}{\mathbb{Z}}
\newcommand{\la}{\langle}
\newcommand{\ra}{\rangle}
\newcommand{\BK}{\mathbb{K}}

\newcommand{\lla}{\langle\!\langle}
\newcommand{\rra}{\rangle\!\rangle}

\newtheorem{proposition}{Proposition}[section]
\newtheorem{theorem}[proposition]{Theorem}
\newtheorem{lemma}[proposition]{Lemma}

\makeatother

\begin{document}
\begin{frontmatter}

\title{Nonintersecting random walks in the neighborhood of a symmetric tacnode}
\runtitle{Tacnode process}

\begin{aug}
\author[A]{\fnms{Mark} \snm{Adler}\thanksref{t1}\ead[label=e1]{adler@brandeis.edu}},
\author[B]{\fnms{Patrik L.} \snm{Ferrari}\thanksref{t2}\ead[label=e2]{ferrari@uni-bonn.de}}
\and
\author[A,C]{\fnms{Pierre} \snm{van Moerbeke}\corref{}\thanksref{t3}\ead[label=e3]{pierre.vanmoerbeke@uclouvain.be}\ead[label=e4]{vanmoerbeke@brandeis.edu}}
\runauthor{M. Adler, P. L. Ferrari and P. van Moerbeke}
\affiliation{Brandeis University, Bonn University, and Universit\'e de
Louvain and~Brandeis~University}
\address[A]{M. Adler\\
P. van Moerbeke\\
Department of Mathematics\\
Brandeis University\\
Waltham, Massachusetts 02453\\
USA\\
\printead{e1}\\
\hphantom{E-mail: }\printead*{e4}} 
\address[B]{P. L. Ferrari\\
Institute for Applied Mathematics\\
Bonn University\\
Endenicher Allee 60\\
53115 Bonn\\
Germany\\
\printead{e2}}
\address[C]{P. van Moerbeke\\
Department of Mathematics\\
Universit\'e de Louvain\\
1348 Louvain-la-Neuve\\
Belgium\\
\printead{e3}}
\end{aug}

\thankstext{t1}{Supported by NSF Grant DMS-07-04271.}
\thankstext{t2}{Supported by the German National Foundation via the
SFB611-A12 project.}
\thankstext{t3}{Supported by NSF Grant DMS-07-04271, a European Science
Foundation grant (MISGAM),
a Marie Curie grant (ENIGMA), FNRS and ``Inter-University Attraction Pole
(Belgium)'' (NOSY) grants.}

\received{\smonth{5} \syear{2011}}
\revised{\smonth{10} \syear{2011}}

%
\begin{abstract}
Consider a continuous time random walk in $\Z$ with independent and
exponentially distributed jumps $\pm1$. The model in this paper
consists in an infinite number of such random walks starting from the
complement of $\{-m,-m+1,\ldots,m-1,m\}$ at time $-t$, returning to the
same starting positions at time $t$, and \textit{conditioned not to
intersect}. This yields a determinantal process, whose gap
probabilities are given by the Fredholm determinant of a kernel. Thus
this model consists of two groups of random walks, which are contained
within two ellipses which, with the choice $m\simeq2t$ to leading
order, just touch: so we have a \textit{tacnode}. We determine the new
limit extended kernel under the scaling $m=\lfloor2t+\sigma
t^{1/3}\rfloor$, where parameter $\sigma$ controls the strength of
interaction between the two groups of random walkers.
\end{abstract}

%

\end{frontmatter}

\section{Introduction}
In the past decade, systems of vicious random walks and
nonintersecting Brownian motions have been investigated, and
quantities such as the correlation functions~\cite{NKT02}, the
one-point distribution functions and limit processes under appropriate
scaling limits have been studied. Nonintersecting Brownian motions
arise in the study of random matrices~\cite{FN98,KT04,KT07b}, and space
(and/or) time discrete versions in random tiling and growth
models~\cite{Jo02b,Jo03b,Jo04,PS02,SI03,FS03,OR01,Fer04}. Most of these works use
the mathematical framework shared by Brownian motions starting from a
point, and either ending at the same point after a given time or the
boundary condition is free (with possible extra boundary conditions
like staying positive~\cite{Nag03,TW07b}).

Consider $N$ nonintersecting Brownian bridges $x_i(\tau)$ on $\R$,
leaving from $0$ at time $\tau=-2N$ and forced to $0$ at time $\tau
=2N$. For large $N$, the mean density of Brownian paths\vspace*{1pt} has support,
for each $-2N<\tau<2N$, on the interval $(-\sqrt{4N^2-\tau^2},\sqrt
{4N^2-\tau^2})$. This means that on the macroscopic scale, where space
and time units are set equal to $N$, one sees a circle. Near its
boundary, the density  of Brownian paths is of order $N^{-1/3}$, thus to
see something nontrivial one needs to look in a space window\vspace*{1pt} of size
$N^{1/3}$ and, by Brownian scalings, a time window of size $N^{2/3}$.
We call this the ``\textit{Airy microscope},'' since it holds
%
\begin{equation}\label{EqIntro1}
\lim_{N\to\infty}\Pb\bigl(\mbox{all }N^{-1/3}\bigl(x_i
({2s}{N^{2/3}})-{2N}\bigr)\in{E^c-s^2} \bigr)= \Pb
\bigl({\cal
A}_2(s)\cap E=\varnothing\bigr),\hspace*{-28pt}
\end{equation}
where ${\cal A}_2$ is the so-called \textit{Airy$_2$ process}. It has a
\textit{universal} character and was discovered in the context of the
so-called multilayer PNG model~\cite{PS02}. The scaling (\ref
{EqIntro1}) is equivalent to the customary $N^{-1/6}$-GUE-edge
rescaling along the circle for nonintersecting Brownian motions
leaving from the origin at time $t=0$ and returning to the origin at
time $t=1$; this is done by an appropriate change of the variance of
the Brownian motions.

In the context of growth models, generalizations have been introduced
with external sources~\cite{BR00,SI04}. Its analog in terms of Brownian
motions is to require that a finite number of Brownian motions end up
at some point $\alpha N$. Then under the scaling in (\ref{EqIntro1}),
the limit process is a transition process from Airy$_2$ to Brownian
motion. For extensions to more general sources, see~\cite{BBP06,BP07},
while for the case that the top $r$ Brownian motion ends up at $2N$,
see~\cite{AvMD08} and~\cite{FAvM08}.

A further known situation occurs when a fraction $pN$ of the $N$
nonintersecting Brownian motions (leaving from the origin at time
$t=-2N$) end at time $t=2N$ at position $a N$ and another fraction
$(1-p)N$ at $b N$, with $a<b$. When $N\to\infty$, the mean density of
Brownian particles has its support on one interval in the beginning and
on two intervals near the end. Thus a bifurcation appears for some
intermediate time $\tau_0$, where one interval splits into two
intervals, creating a ``heart-like'' shape with a cusp at the origin.
Near this cusp appears a new \textit{universal} process, upon looking
through the ``\textit{Pearcey microscope},'' where the space window is
$N^{1/4}$, and the time window is $N^{1/2}$. The new process is called
the \textit{Pearcey process}~\cite{TW06} and is independent of the values
of $a, b$ and $p$; see~\cite{AOvM10}. Once the bifurcation has taken
place, the Brownian motions will eventually fluctuate like the Airy$_2$
process near the edge, with a transition from the Pearcey to the
Airy$_2$ process~\cite{ACvM10}. The Pearcey process has also been
obtained as the limit of discrete models; see~\cite{OR07,BK09,BD10}.

The motivation of our work is to understand what happens when half of
the nonintersecting Brownian motions start and end at a point, while
the second half start and end at another point.
When the two starting points are sufficiently far apart from each
other, the mean density of particles will be confined to two separate
circles, with Airy$_2$ processes appearing near the boundary, as
described above. When the two starting points move away from each other
at an appropriate rate proportional to $N$, the two circles will just
touch, creating a tacnode. A new \textit{critical} process appears by
looking at the two sets of nonintersecting Brownian motions, which
experience a brief meeting in the neighborhood of the tacnode, but
looked at with the \textit{Airy scaling}; we call it the
\textit{tacnode process}. Pictorially it can be thought of as two Airy$_2$
processes touching; see Figure~\ref{FigTacnode}.%

%
\begin{figure}

\includegraphics{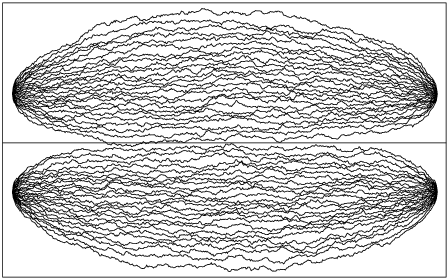}

\caption{Illustration of the tacnode with $N=50$ Brownian bridges.}
\label{FigTacnode}
\end{figure}

In this paper we obtain an explicit formula for the kernel governing
this tacnode process for nonintersecting continuous-time
random walks, rather than nonintersecting Brownian motions. The same
result is expected to hold for the Brownian motion case, since under
the scaling the discrete nature of the random walks is lost, and the
random walks become Brownian paths. Our main result is the limiting
kernel at the tacnode under appropriate scaling limit, stated in
Theorem~\ref{ThmExtKernelAsympt}. Before taking the limit, the kernel
is given by Theorem~\ref{MainTheorem}. The model is to let two groups
of nonintersecting random walks with jumps $\pm1$, rate~$1$ and
$2m+1$ integers apart evolve during a total time or order $m$, with
space--time rescaled \textit{\`a la Airy}, namely $x\sim\xi m^{1/3}$ and
$\tau\sim s m^{2/3}$ as suggested by formula~(\ref{EqIntro1}). The
parameter $m$, defined here, plays the role of the number of particles
$N$, previously defined.

There is an important difference with respect to the Airy$_2$ and
Pearcey cases: here we have a one-parameter family of
processes, which is obtained by modulating the endpoints' distance
between the two sets of Brownian motions over distance of order
$N^{1/3}$. For the Pearcey processes (and the Airy$_2$ process),
geometric changes of this type only have the effect of modifying the
position (and orientation) of the cusp, but the underlying Pearcey
process remains unchanged. In the literature there is another known
situation with a process in a tacnode-like geometry~\cite{BD10}, which,
however, differs from the present one.


For Brownian motions the problem can be approached using multiple
orthogonal polynomials~\cite{DK08}; then Delvaux, Kuijlaars and Zhang~\cite{DKZ10}
carry out asymptotics for these polynomials yielding a Riemann--Hilbert
description of the tacnode process kernel (which meanwhile appeared on
the arXiv). In the forthcoming paper~\cite{Joh10}, Johansson uses a different
approach leading to an explicit kernel for the Brownian motion problem,
but seemingly and surprisingly different from the one obtained in the
present paper. In another forthcoming paper Adler, Johansson and van
Moerbeke~\cite{AJvM11} consider a tacnode process in the context of domino
tilings of two overlapping Aztec diamonds and found yet another kernel;
in the same paper they show that the kernels obtained are all
equivalent! A direct relation with the Riemann--Hilbert type
formulation of the kernel~\cite{DKZ10} remains an open problem. In a recent
preprint about nonintersecting Brownian motions, Ferrari and Veto~\cite{FV12}
discuss a kernel for a nonsymmetric tacnode, which contains a
parameter sensing the relative number of Brownian motions, or
equivalently, the ratio of the curvature of the curves meeting at the
tacnode.

\subsection*{Outline}
In Section~\ref{Model} we define the model and state the two main
results. In Section~\ref{SectFiniteSyst}, Theorem~\ref{thmKernelPart},
we derive the finite time result for $\tau=0$, which is reshaped in
Section~\ref{shape} as a preparation to carrying out the large time
limit. Before actually doing this, we indicate in Section~\ref{s5} how
to introduce the time, leading to the finite multi-time kernel in
Theorem~\ref{ThmExtKernel}, an extension of the kernel appearing in
Proposition~\ref{propDecomp}. In Section~\ref{SectAsymptotics}, we
take the limit of the multi-time kernel, leading to the proof of the
first formula of Theorem~\ref{ThmExtKernelAsympt}. In Section \ref
{IntegralReprKernel}, we sketch the proof of the double integral
representation of the kernel, the second formula of Theorem~\ref
{ThmExtKernelAsympt}, using the steepest descent analysis.

\section{Model and results}\label{Model}
Consider a continuous time random walk in $\Z$ with jumps $\pm1$,
occurring independently with rate $1$; that is, the waiting times of
the up- and down-jumps are independent and exponentially distributed
with mean $1$. The transition probability $p_t(x,y)$ of going from $x$
to $y$ during a time interval of length $t$ is given by
%
\begin{equation} \label{tr-pr0}
p_t(x,y) = e^{-2t} I_{|x-y|}(2t),
\end{equation}
where $I_n$ is the modified Bessel function of degree $n$; see~\cite{AS84}.

Consider now an infinite number of continuous time random walks
starting from $\{\ldots,-m-2,-m-1\}\cup\{m+1,m+2,\ldots\}$ at
time $\tau=-t$, returning to the starting~positions at time $\tau=t$,
and \textit{conditioned not to intersect}; see Figure~\ref{figTrajectoriesB}.
%
\begin{figure}

\includegraphics{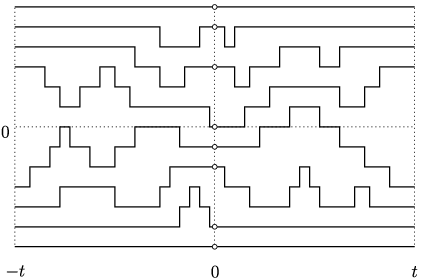}

\caption{The lines are the nonintersecting walks $\tilde{\mathbf{x}}$.
The white circles are the support of the point process $\tilde\eta$.}
\label{figTrajectoriesB}
\end{figure}
Denote $\tilde x_k(\tau)$ the position of the walk that starts and ends
at position $k$. Then, the point process $\tilde\eta$ on $\Z$
(described by the little white circles in Figure \ref
{figTrajectoriesB}) defined by
%
\begin{equation}
\tilde\eta(x)=\sum_{k\in\Z\setminus\{-m,\ldots,m\}} \delta
_{x,\tilde x_k(0)}
\end{equation}
with $\delta$ the Kronecker-delta, is determinantal; that is, there
exists a kernel $\widetilde\BK_{m}$ such that the $k$-point correlation
function $\rho^{(k)}$ is given by
$\rho^{(k)}(y_1,\ldots,y_k)=\det(\widetilde\BK
_{m}(y_i,y_j))_{1\leq
i,j\leq k}$. One of the interesting quantities is the \textit{gap
probability of a set $E$}, which is given by $\Pb(\tilde
\eta(\Id_E)=0)$, that is, the probability that \textit{none of the random
walks are in $E$} at time $\tau=0$. For a determinantal point process
the gap probability is given by the Fredholm determinant of the
associated kernel $\widetilde\BK_{m}$ projected onto $E$. For more
informations on determinantal point processes,
see~\cite{Lyo03,BKPV05,Sos06,Jo05,Spo05}.

The determinantal structure still holds if we consider the point
process on a set of time-slices instead of a single time $\tau=0$. This
means that given times $\tau_1<\tau_2<\cdots<\tau_p$ in the interval
$(-t,t)$, the point process on $\{\tau_1,\ldots,\tau_p\}\times\Z$
defined by
%
\begin{equation}
\tilde\eta(\tau,x)=\sum_{r=1}^p\sum_{k\in\Z\setminus\{-m,\ldots
,m\}}
\delta_{(\tau,x),(\tau_r,\tilde x_k(\tau_r))}
\end{equation}
is determinantal. That is, the space--time correlation functions are
given by the determinant of an extended kernel, which we denote by
$\widetilde\BK_m^{\mathrm{ext}}(t_1,x_1;t_2,x_2)$, where $t_i\in\{\tau
_1,\ldots,\tau_p\}$ and $x_i\in\Z$.

%
\begin{figure}

\includegraphics{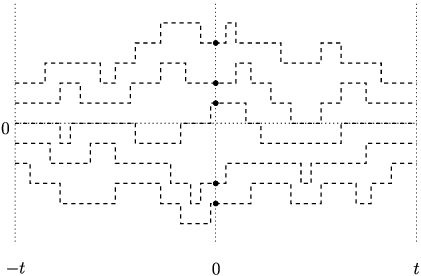}

\caption{The dotted lines are the nonintersecting
walks $\mathbf{x}$, the dual process of $\tilde{\mathbf{x}}$ of
Figure~\protect\ref{figTrajectoriesB}. The black circles are the
support of the
point process $\eta$.} \label{figTrajectories}
\end{figure}

It is more convenient to first study the \textit{dual} or complementary
process $ {\mathbf{x}} (\tau)$. The \textit{dual proceeds along the
gaps of }$\tilde{\mathbf{x}} (\tau)$. In this instance, the dual $
{\mathbf{x}} (\tau)$ of $\tilde{\mathbf{x}} (\tau)$ is described by
$n=2m+1$ ($m\in\N$) \textit{nonintersecting} continuous-time random
walks, starting from $-m,-m+1,\ldots,m-1,m$ at time $\tau=-t$,
returning to the starting positions at time $\tau=t$; see Figure
\ref{figTrajectories}, and Figure~\ref{figTrajectoriesC} for the
superposition of the trajectories of ${\mathbf{x}}(\tau)$ and
$\tilde{\mathbf{x}}(\tau)$.

\begin{figure}[b]

\includegraphics{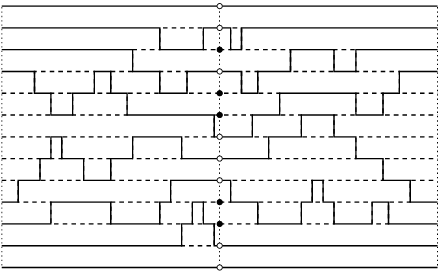}

\caption{Superposition of Figures \protect\ref{figTrajectoriesB} and
\protect\ref{figTrajectories}.}
\label{figTrajectoriesC}
\end{figure}

In particular, the dual process ${\mathbf{x}}(\tau)$ at $\tau=0$ is
given by the little black circles in Figure~\ref{figTrajectories}. The
probability measure at time $\tau=0$ is obtained by the
Karlin--McGregor formula~\cite{KM59}, and thus it is a determinantal
process for a kernel $\BK_m$. Finally, the complementation principle by
Borodin, Olshanski and Okounkov (see Appendix of~\cite{BOO00}) tells us
that, if the kernel $\BK_m$ governs the process $ {\mathbf{x}} (\tau)$,
then the kernel $\widetilde\BK_m=\Id-\BK_m$ describes the dual process~$\tilde{\mathbf{x}} (\tau)$.
%
\begin{theorem}\label{MainTheorem}
The determinantal point process $\tilde\eta(\tau,x)$ on $\{\tau
_1,\ldots,\tau_p\}\times\R$, $\tau_i\in(-t,t)$, defined by the two
groups of nonintersecting walkers, starting and ending $2m+1$ apart,
at times $-t$ and $t$, respectively, has gap probabilities on any
compact set $E\subset\{\tau_1,\ldots,\tau_p\}\times\R$ given by
%
\begin{equation}
\Pb\bigl(\tilde\eta(\Id_E)=0\bigr)=\det(\Id-\widetilde\BK_m^{\mathrm{ext}})_{L^2(E)},
\end{equation}
where the kernel $\widetilde\BK_m^{\mathrm{ext}}$ is given by
%
\begin{eqnarray}\label{eqThmKernelHolesIntroExt}
&&
\frac{e^{2t_2}}{e^{2t_1}}\widetilde\BK_m^{\mathrm{ext}}(t_1,x_1;t_2,x_2)\nonumber\\
&&\qquad=-\Id_{[t_2<t_1]}I_{|x_1-x_2|}\bigl(2(t_2-t_1)\bigr)\nonumber\\
&&\qquad\quad{}-\frac{V_m}{(2\pi\I)^2} \oint_{\Gamma_0}dz\oint_{\Gamma_{0,z}}dw
\frac{e^{t(z-z^{-1})}}{e^{t(w-w^{-1})}}
\frac{e^{-t_1(z+z^{-1})}}{e^{-t_2(w+w^{-1})}}
\frac{w^{x_2-m-1}}{z^{x_1-m}}\nonumber\\
&&\qquad\quad\hspace*{110pt}{}\times\frac{H_{2m+1}(w)
H_{2m+1}(z^{-1})}{z-w}\\
&&\qquad\quad{} - \frac{V_m}{(2\pi\I)^2} \oint_{\Gamma_0}dw\oint_{\Gamma_{0,w}}dz
\frac{e^{t(w-w^{-1})}}{e^{t(z-z^{-1})}}
\frac{e^{-t_1(z+z^{-1})}}{e^{-t_2(w+w^{-1})}}
\frac{w^{x_2+m}}{z^{x_1+m+1}}\nonumber\\
&&\qquad\quad\hspace*{112pt}{}\times\frac{H_{2m+1}(z)H_{2m+1}(w^{-1})}{w-z}\nonumber\\
&&\qquad\quad{} -\Id_{[x_1\neq x_2]}\frac{V_m}{2\pi\I}\oint_{\Gamma_0}dz
\frac{e^{(t_2-t_1)(z+z^{-1})}}{z^{x_1-x_2+1}}
H_{2m+1}(z^{-1})H_{2m+1}(z)\nonumber
\end{eqnarray}
with $V_m:=1/(H_{2m+1}(0)H_{2m+2}(0))$. The function $H_n$ is itself
the Fredholm determinant on $\ell^2(\{n,n+1,\ldots\})$
%
\begin{equation}
H_n(z^{-1}):=\det\bigl(\Id-K(z^{-1})\bigr)_{\ell^2(\{n,n+1,\ldots\})}
\end{equation}
of the kernel
%
\begin{equation}
K(z^{-1})_{k,\ell}:=\frac{(-1)^{k+\ell}}{(2\pi\I)^2}\oint_{\Gamma
_0}du \oint_{\Gamma_{0,u}}dv\frac{u^{\ell}}{v^{k+1}} \frac
{1}{v-u}\frac
{u-z}{v-z}\frac{e^{2t(u-u^{-1})}}{e^{2t(v-v^{-1})}},
\end{equation}
where $\Gamma_{0}$ is any anticlockwise simple loop enclosing $0$ and
similarly $\Gamma_{0,u}$ encircles the poles at $0$ and $u$ (but
not $z$).\setcounter{footnote}{3}\footnote{For any set of points $S$, the notation $\oint
_{\Gamma
_{S}} dz f(z)$ means that the integration path goes anticlockwise
around the points in $S$ but does not include any other poles of $f$.}

The extended kernel, governing the process $\tilde\eta(\tau,x)$, is
given in terms of the kernel $\widetilde\BK_m(x_1,x_2)=\widetilde\BK
_m^{\mathrm{ext}}(0,x_1;0,x_2)$, governing the distribution $\tilde\eta
(0,x)$, by
%
\begin{eqnarray}\label{27}
\widetilde\BK_m^{\mathrm{ext}}(t_1,x_1;t_2,x_2) &=&-\Id_{[t_2<t_1]}
\bigl(e^{(t_2-t_1) {\cal H}}\bigr)(x_1,x_2)\nonumber\\[-8pt]\\[-8pt]
&&{}+(e^{-t_1 \cal H}\widetilde\BK_m e^{t_2 \cal H})(x_1,x_2),\nonumber
\end{eqnarray}
where ${\cal H}$ is the discrete Laplacian
%
\begin{equation}\label{Laplacian0}
({\cal H} f)(x)= f(x+1)+f(x-1)-2f(x).
\end{equation}
\end{theorem}

Remark that the transition probability of (\ref{tr-pr0}), defined for
$t\geq0$, can be written as $p_t(x,y)= e^{t{\cal H}}\Id
(x,y)=:e^{t{\cal H}} (x,y)$. Here, $\Id$ denotes the identity operator
on $\Z$, that is, $\Id(x,y)=1$ if $x=y$ and $\Id(x,y)=0$ if $x\neq y$.

The formula for the kernel $\widetilde\BK_m(x_1,x_2)=\widetilde\BK
_m^{\mathrm{ext}}(0,x_1;0,x_2)$ at $t_1=t_2=0$ of Theorem~\ref{MainTheorem},
will be established in Section~\ref{SectFiniteSyst}, whereas the one
for $\widetilde\BK^{\mathrm{ext}}_m$ will be shown in Section~\ref{s5}. In
Sections~\ref{shape} and~\ref{s5}, it will be shown that both kernels
$\widetilde\BK_m(x,y)$ and $\widetilde\BK^{\mathrm{ext}}_m(t_1,x_2;t_2,x_2)$ have a representation, whose constituents can
be expressed in terms of Bessel functions; see the expression (\ref
{eqpropDecomp}) and the time-dependent kernel~(\ref{Kext}), derived
from (\ref{eqpropDecomp}), via recipe (\ref{27}).
Also, note that the kernel $K(z^{-1})$ is a rank-one perturbation of
the kernel $K(0)$, whose Fredholm determinant
%
\begin{equation}\label{eq313}
H_n(0)=\det\bigl(\Id-K(0)\bigr)_{\ell^2(\{n,n+1,\ldots\})}
\end{equation}
is the distribution of the longest increasing subsequence of a random
permutation in the Poissonized version, or, equivalently, it
yields the distribution of the
height function in the polynuclear growth (PNG) model~\cite{BDJ99,PS02}.
In the scaling limit, considered in Section~\ref{SectAsymptotics},
$H_n(0)$ will converge to the Tracy--Widom distribution $F_2$.

To study the limiting behavior, when $m,t\to\infty$, consider first
the system of nonintersecting random walks starting at time $-t$ and
ending at positions $\{\ldots,-m-2,-m-1\}$ at time $t$. This is,
up to a shift by $m+1$, the multilayer PNG model studied by Pr\"ahofer
and Spohn in~\cite{PS02}. Their work\vspace*{1pt} shows that the top random walk at
time $\tau=0$ has fluctuations around $x=-m+2t$ of order $t^{1/3}$. By
symmetry, if one considers only the nonintersecting random walks
starting and ending at position $\{m+1,m+2,\ldots\}$, the bottom random
walk at time $\tau=0$ fluctuates around $x=m-2t$ also in the spatial
scale $t^{1/3}$.

The top and bottom random walks interact if the proportion of deleted
configurations, due to interaction, is nonzero. This happens when
$m=2t$ to leading order in $t$. The first scaling where interaction is
relevant is given by $m=2t+\sigma t^{1/3}$. The parameter $\sigma$
modulates the strength of interaction of the two sets of
nonintersecting random walks. In the extreme cases $\sigma\to\infty$,
we clearly (by a simple probabilistic argument) go back to the
situation of two independent PNG models; thus the top of the lower
walks and the bottom of the upper walks are governed by the Airy$_2$
process~\cite{PS02}. On the other hand, when $\sigma\to-\infty$, one
expects to see a point process governed by the sine kernel or the
Pearcey process. Moreover, locally the paths will looks like random
walks, so the exponents in the scaling for time and space are in a
ratio $2:1$. Thus, we set the scaling\footnote{We do not write explicitly
the integer parts, since in the $t\to\infty$ limit it is irrelevant.}
%
\begin{equation}\label{eqScalingSpaceTime}
m=2t+\sigma t^{1/3},\qquad x_i=\xi_i t^{1/3},\qquad t_i=s_i
t^{2/3},\qquad i=1,2.
\end{equation}
Also note that for each time $-t<\tau<t$, the density of particles has
its support on two semi-infinite intervals, whose boundary, as a
function of~$\tau$, describes two curves, which at $\tau=0$ form a
\textit{tacnode}. The purpose of Theorem~\ref{ThmExtKernelAsympt} is to
describe the fluctuations of the random walks in the $t\to\infty$
limit in the neighborhood of $(x,\tau)=(0,0)$, but in the new
space--time scale, given by (\ref{eqScalingSpaceTime}).

In order to state the second main result, define the standard Airy kernel,
%
\begin{equation}\label{AiryK}
K_{\mathrm{Ai}}(\xi_1,\xi_2):=\int_0^\infty d\lambda \Ai(\xi
_1+\lambda)\Ai
(\xi_2+\lambda)
\end{equation}
and the function ${\cal Q}(\kappa)$, already appearing in~\cite{TW94},
%
\begin{equation}\label{Q}
{\cal Q}(\kappa):=[(\Id-\chi_{\tilde\sigma} K_{\mathrm{Ai}} \chi
_{\tilde
\sigma})^{-1}\chi_{\tilde\sigma} {\Ai}](\kappa)\qquad \mbox{with
}\tilde
\sigma:=2^{2/3}\sigma
\end{equation}
and where $\chi_{a}(x)=\Id_{[x>a]}$. We further set
%
\begin{equation}\label{Airy}
\Ai^{(s)} (\xi):= e^{\xi s+(2/3) s^3}\Ai(\xi+ s^2),
\end{equation}
which equals the standard Airy function $\Ai(\xi)$, when $s=0$, and
define the functions
%
\begin{eqnarray}\label{eqSpaceTimeAB}\quad
{\cal A}(s,\xi)&:=& \Ai^{(s)}(\sigma-\xi)\nonumber\\
&&{}+\int_{\tilde\sigma
}^\infty
d\kappa\int_0^\infty d\alpha\,
{\cal
Q}(\kappa) \Ai(\kappa+\alpha) \Ai^{(s)}(2^{1/3}\alpha+\sigma-\xi
),\\
{\cal B}(s,\xi)&:=& \int_{\tilde\sigma}^\infty d\kappa\,{\cal
Q}(\kappa)\Ai^{(s)}(2^{1/3}\kappa-\sigma+\xi)\nonumber
\end{eqnarray}
and
%
\begin{eqnarray}\label{eqSpaceTimeC}
{\cal C}(s,\xi)&:=&2^{-1/3}\int_{\tilde\sigma}^\infty d\kappa\,{\cal
Q}(\kappa) \biggl[ \Ai^{(2^{-2/3}s)}(\kappa+2^{-1/3}\xi)\nonumber\\
&&\hspace*{87pt}{}+ \int_{\tilde\sigma}^\infty d\lambda\,{\cal Q}(\lambda) \int
_0^\infty
d\alpha \Ai(\alpha+\lambda) \nonumber\\[-8pt]\\[-8pt]
&&\hspace*{174pt}{}\times\Ai^{(2^{-2/3}s)}(\alpha+\kappa
+2^{-1/3}\xi)\biggr]\nonumber\\
&&{}+(\xi\leftrightarrow-\xi),
\nonumber
\end{eqnarray}
where with $(\xi\leftrightarrow-\xi)$ we mean the same expression with
$\xi$ replaced by $-\xi$.
Finally, we define the following two Laplace transforms, $ \hat{\cal
P}(u)$ and $ \hat{\cal Q}(u)$:
%
\begin{eqnarray}\label{Laplace}
\hat{\cal Q}(u ) &:=& \int^{\infty}_{\tilde\sigma} d\kappa\, {\cal
Q}(\kappa) e^{\kappa u 2^{1/3}},\nonumber\\[-8pt]\\[-8pt]
\hat{\cal P}(u)&:=&- \int_0^\infty d\kappa\, e^{-\kappa u 2^{1/3}}
\int
^{\infty}_{\tilde\sigma} d\mu\, {\cal Q}(\mu) \Ai(\mu+\kappa).
\nonumber
\end{eqnarray}
%
\begin{theorem}\label{ThmExtKernelAsympt}
Near the tacnode appears a new determinantal process on $\{s_1,\ldots
,s_p\}\times\R$, the tacnode process ${\cal T}$, whose gap
probabilities on any compact set $E\subset\{s_1,\ldots,s_p\}\times\R$
are given by
%
\begin{equation}
\Pb\bigl({\cal T}(\Id_E)=0\bigr)=\det(\Id-{\cal K}^{\mathrm{ext}})_{L^2(E)}.
\end{equation}
The kernel ${\cal K}^{\mathrm{ext}}$ is the limit of $\widetilde\BK_m^{\mathrm{ext}}$ under the scaling (\ref{eqScalingSpaceTime}),
%
\begin{equation}
{\cal K}^{\mathrm{ext}}(s_1,\xi_1;s_2,\xi_2):=\lim_{t\to\infty} \frac
{(-1)^{x_2} e^{4t_2}}{(-1)^{x_1} e^{4t_1}} t^{1/3} \widetilde\BK
_m^{\mathrm{ext}}(t_1,x_1;t_2,x_2),
\end{equation}
where the convergence is uniform for $\xi_1, \xi_2$ and $s_1,s_2$ in
bounded sets. The kernel ${\cal K}^{\mathrm{ext}}$ has the following
representations:
%
\begin{eqnarray}\label{ExtKernelA}
\hspace*{-4pt}&&
{\fontsize{10.15pt}{11pt}\selectfont{\mbox{$\displaystyle {\cal
K}^{\mathrm{ext}}(s_1,\xi_1;s_2,\xi_2)$}}}
\nonumber\\
\hspace*{-4pt}&&
{\fontsize{10.15pt}{11pt}\selectfont{\mbox{$\displaystyle \qquad= -\frac{\Id_{[s_2<s_1]}}{\sqrt{4\pi(s_1-s_2)}} \exp\biggl(-\frac
{(\xi
_1-\xi_2)^2}{4(s_1-s_2)}\biggr)
+{\cal C}(s_1-s_2,\xi_1-\xi_2)$}}}\nonumber\\
\hspace*{-4pt}&&
{\fontsize{10.15pt}{11pt}\selectfont{\mbox{$\displaystyle \qquad\quad{}+ \int_0^\infty d\gamma\bigl({\cal A}(s_1,\xi_1-\gamma){\cal
A}(-s_2,\xi_2-\gamma)
+{\cal A}(s_1,-\xi_1-\gamma){\cal A}(-s_2,-\xi
_2-\gamma)$}}} \nonumber\\[-8pt]\\[-8pt]
\hspace*{-4pt}&&
{\fontsize{10.15pt}{11pt}\selectfont{\mbox{$\displaystyle \qquad\quad\hspace*{44.5pt}{}-{\cal A}(s_1,\xi_1-\gamma){\cal
B}(-s_2,\xi_2-\gamma)
-{\cal
A}(s_1,-\xi_1-\gamma){\cal B}(-s_2,-\xi_2-\gamma)$}}}\nonumber\\
\hspace*{-4pt}&&
{\fontsize{10.15pt}{11pt}\selectfont{\mbox{$\displaystyle \qquad\quad\hspace*{44.5pt}{}-{\cal B}(s_1,\xi_1-\gamma){\cal A}(-s_2,\xi_2-\gamma)
-{\cal
B}(s_1,-\xi_1-\gamma){\cal A}(-s_2,-\xi_2-\gamma)\bigr)$}}}\nonumber\\
\hspace*{-4pt}&&
{\fontsize{10.15pt}{11pt}\selectfont{\mbox{$\displaystyle \qquad\quad{}-\int_{-\infty}^0 d\gamma\bigl({\cal B}(s_1,\xi_1-\gamma){\cal
B}(-s_2,\xi_2-\gamma)
+{\cal B}(s_1,-\xi_1-\gamma){\cal
B}(-s_2,-\xi_2-\gamma)\bigr)$}}}\nonumber
\end{eqnarray}
as well as (with arbitrary $\delta>0$)
%
\begin{eqnarray}\label{ExtKernelB}
&&{\cal K}^{\mathrm{ext}}(s_1,\xi_1;s_2,\xi_2)\nonumber\\
&&\qquad= -\frac{\Id_{[s_2<s_1]}}{\sqrt{4\pi(s_1-s_2)}} \exp\biggl(-\frac
{(\xi
_1-\xi_2)^2}{4(s_1-s_2)}\biggr)
+{\cal C}(s_1-s_2,\xi_1-\xi_2)\nonumber\\
&&\qquad\quad{}+\frac{1}{(2\pi\I)^2}\int_{\delta+\I\R} du\int
_{-\delta
+\I\R} dv
\frac{e^{{u^3}/3-\sigma u}}{e^{{v^3}/3-\sigma v}}\frac{e^{s_1
u^2}}{e^{s_2 v^2}}
\biggl(\frac{e^{\xi_1 u}}{e^{\xi_2 v}}+\frac{e^{-\xi_1 u}}{e^{-\xi_2
v}}\biggr)\nonumber\\
&&\hphantom{dv}\hspace*{116pt}\qquad\quad{}\times\frac{(1-\hat{\cal P}(u))(1-\hat{\cal P}(-v))}{u-v}\nonumber\\
&&\qquad\quad{}-\frac{1}{(2\pi\I)^2}\int_{2\delta+\I\R} du \int
_{\delta+\I\R} dv
\frac{e^{{u^3}/3-\sigma u}}{e^{-{v^3}/3-\sigma v}}\frac{e^{s_1
u^2}}{e^{s_2 v^2}}
\biggl(\frac{e^{\xi_1 u}}{e^{\xi_2 v}}+\frac{e^{-\xi_1 u}}{e^{-\xi_2
v}}\biggr)\nonumber\\[-8pt]\\[-8pt]
&&\hphantom{dv}\hspace*{113pt}\qquad\quad{}\times
\frac{(1-\hat{\cal P}(u)) \hat{\cal Q}(-v)}{u-v}\nonumber\\
&&\qquad\quad{}-\frac{1}{(2\pi\I)^2}\int_{-\delta+\I\R} du\int
_{-2\delta+\I\R} dv
\frac{e^{-{u^3}/3-\sigma u}}{e^{{v^3}/3-\sigma v}} \frac{e^{s_1
u^2}}{e^{s_2 v^2}}
\biggl(\frac{e^{\xi_1 u}}{e^{\xi_2 v}}+\frac{e^{-\xi_1 u}}{e^{-\xi_2
v}}\biggr)\nonumber\\
&&\hphantom{dv}\hspace*{127pt}\qquad\quad{}\times
\frac{(1-\hat{\cal P}(-v)) \hat{\cal Q}(u)}{u-v}\nonumber\\
&&\qquad\quad{}+\frac{1 }{(2\pi\I)^2}\int_{-\delta+\I\R} du \int
_{\delta+\I\R} dv
\frac{e^{-{u^3}/3-\sigma u}}{e^{-{v^3}/3-\sigma v}}
\frac{e^{s_1 u^2}}{e^{s_2 v^2}}
\biggl(\frac{e^{\xi_1 u}}{e^{\xi_2 v}}+\frac{e^{-\xi_1 u}}{e^{-\xi_2
v}}\biggr)\nonumber\\
&&\hphantom{dv}\hspace*{116pt}\qquad\quad{}\times
\frac{\hat{\cal Q}(u) \hat{\cal Q}(-v)}{u-v}.\nonumber
\end{eqnarray}
\end{theorem}

Note the kernel (\ref{ExtKernelB}) is invariant under the involution
$(s_1,\xi_1; s_2,\xi_2) \mapsto (-s_2,-\xi_2; -s_1,-\xi_1)$, thus
reflecting the symmetry of the symmetric tacnode.

The form (\ref{ExtKernelA}) of the limiting extended kernel in
Theorem~\ref{ThmExtKernelAsympt} will be shown in Section \ref
{SectAsymptotics}, whereas a sketch of the proof of its double integral
representation (\ref{ExtKernelB}) will be given in Section \ref
{IntegralReprKernel}.

In the preprint~\cite{DKZ10}, the analogous problem for Brownian Motion
will be analyzed with the Riemann--Hilbert approach applied to multiple
orthogonal polynomials. It would be interesting to see how to relate
the two formulas (which we expect to be equivalent).

\section{\texorpdfstring{Finite system at $\tau=0$}{Finite system at tau=0}}\label{SectFiniteSyst}
In this section we will prove Theorem~\ref{MainTheorem}, in particular
the formula for kernel $\widetilde{\BK}_m(x,y)=\widetilde\BK_m^{\mathrm{ext}}(0,x;0,y)$, as in (\ref{eqThmKernelHolesIntroExt}), for
$t_1=t_2=0$. Consider a continuous time random  walk in $\Z$ with jumps
$\pm1$, which occur independently with rate $1$; that is, the waiting
times of the up- and down-jumps are independent and exponentially
distributed with mean~$1$. Thus, the number of up-jumps (and similarly
down-jumps) during the time interval $[0,t]$ is Poisson distributed,
%
\begin{equation}
\Pb(k \mbox{ up-jumps during } [0,t])=e^{-t}\frac{t^k}{k!}.
\end{equation}
As will be shown, the transition probability $p_t(x,y)$ of going from
$x$ to $y$ during a time interval of length $t$ is given by
%
\begin{equation}\label{tr-pr}
p_t(x,y) = e^{-2t} I_{|x-y|}(2t),
\end{equation}
where $I_n$ is the modified Bessel function of degree $n$; see \cite
{AS84}. To prove (\ref{tr-pr}), first notice that by symmetry, it is
enough to consider $y-x\geq0$. To go from $x$ to $y$, the process must
perform $k$ steps down and $k+y-x$ steps up. Since the moment, at which
the down or up steps occur, is independent of whether it is a down or
an up step, one may assume the process doing first $k$ steps down and
then $k+y-x$ steps up. By the strong Markov property of the random walk
and the independence of the jumps,
%
\begin{eqnarray}\label{tr-prk}
p_t(x,y)&=&\sum_{k\geq0}\Pb(\{
\mbox{$k+y-x$ up-steps and $k$ down-steps}\}
\mbox{ during time $t$})\nonumber\\
&=& e^{-2t} \sum^{\infty}_{k\geq0}\frac{t^k}{k!}\frac{t^{y-x
+k}}{(y-x+k)!}\\
&=&e^{-2t}
I_{|x-y|}(2t).\nonumber
\end{eqnarray}
The modified Bessel function has the following expressions (for $n\in
\Z$)
%
\begin{equation}\label{eqbesselI}
I_n(2t) =\frac{1}{2\pi\I}\oint_{S^1}\frac
{dz}{z}e^{t(z+z^{-1})}z^{\pm
n} =\sum^{\infty}_{k=0}\frac{t^k}{k!}\frac{t^{k+|n|}}{(k+|n|)!}
\end{equation}
with $S^1=\{z\in\C| |z|=1\}$.

Consider now $n=2m+1$ ($m\in\N$) continuous time random walks starting
from $-m,-m+1,\ldots,m-1,m$ at time $\tau=-t$, returning at the
starting positions at time $\tau=t$, and \textit{conditioned not to
intersect}. Denote by $x_k(\tau)$ the position at time $\tau$ of the
random walk which started from $m+1-k$ (i.e., the $k$th highest one),
see Figure~\ref{figTrajectories} for an illustration with $m=2$.

The probability at time $\tau=0$ is easily obtained by the
Karlin--McGregor formula~\cite{KM59}, namely
%
\begin{eqnarray}\label{1}
&& \Pb\Biggl(\bigcap_{k=1}^{2m+1}\{x_k(0)=y_k\}\bigg| \bigcap
_{k=1}^{2m+1}\{x_k(t)=x_k(-t)=m+1-k\}\Biggr)\nonumber\\
&&\qquad=\mathrm{const}\times\det[p_t(m+1-i,y_j) ]_{1\leq i,j\leq
2m+1}\\
&&\qquad\quad{}\times\det[p_t(y_i,m+1-j) ]_{1\leq i,j\leq2m+1}\nonumber\\
&&\qquad=\mathrm{const}\times(\det[I_{y_i+j-1-m}(2t) ]_{1\leq
i,j\leq
2m+1})^2.\nonumber
\end{eqnarray}
It is well known by~\cite{Bor98} that the process above
%
\begin{equation}
\mathbf{x}(\tau):=\{x_k(\tau),1\leq k\leq2m+1\},\qquad
\tau\in[-t,t],
\end{equation}
with a measure of this form, gives rise to a determinantal point
process (random point measure)
%
\begin{equation}\label{eta}
\eta=\sum_{k=1}^{2m+1}\delta_{x_k(0)}
\end{equation}
with a certain kernel $\BK_m(x,y)$, to be computed in Theorem
\ref{thmKernelPart}.

Instead of the process $\mathbf{x}(\tau)$,
we shall analyze its complementary (dual) process, which we denote by
%
\begin{equation}
\tilde{\mathbf{x}}(\tau)=\{\tilde x_k(\tau),k\in\Z\setminus[1,
2m+1]\}, \qquad\tau\in[-t,t].
\end{equation}
If $\mathbf{x}$ denotes the trajectories of the $2m+1$ particles, then
let $\tilde{\mathbf{x}}$ denote the trajectories of the holes, obtained
by the particle-hole transformation; see Figures~\ref{figTrajectoriesB}
and~\ref{figTrajectoriesC}.

The reason for starting with the process $\mathbf{x}$ is that the
Karlin--McGregor formula applies to a finite number of paths, while
$\tilde{\mathbf{x}}$\vadjust{\goodbreak} has an infinite number of paths. By the
complementation principle in the Appendix of~\cite{BOO00}, the dual
point process at $\tau=0$,
%
\begin{equation}\label{tilde-eta}
\tilde\eta=\sum_{k}\delta_{\tilde x_k(0)},
\end{equation}
is also determinantal with correlation kernel
%
\begin{equation}\label{kernel1b}
\widetilde\BK_m(x,y)=\delta_{x,y}-\BK_m(x,y).
\end{equation}
First of all, we compute the kernel $\BK_m(x,y)$ in a form which will
be suitable for asymptotic analysis.
%
\begin{theorem}\label{thmKernelPart}
The point processes $\eta$ and $\tilde\eta$, defined in (\ref{eta})
and (\ref{tilde-eta}), are determinantal with correlation kernel $\BK
_m$ and $\widetilde\BK_m$ given below. Thus, for any finite subset
$E\subset\Z$, the gap probability of $E$ is given by
%
\begin{eqnarray}
\Pb\bigl(\eta(\Id_E)=0\bigr)&=&\det(\Id-\BK_{m})_{\ell
^2(E)},\nonumber\\[-8pt]\\[-8pt]
\Pb \bigl(\tilde\eta(\Id_E)=0\bigr)&=&\det(\Id- \widetilde\BK
_{m})_{\ell^2(E)}\nonumber
\end{eqnarray}
with kernels $\BK_m(x,y)$ and $\widetilde\BK_m(x,y)$,
invariant\footnote
{As it should from the geometry of the problem! The involution
interchanges the two double integrals in (\ref{eqThmKernelPart}), as is
seen from renaming $w\leftrightarrow z $ in the second double integral;
also the third term, the single integral, only depends on $|x-y|$, as
is seen from $z\to z^{-1}$.} under the involution
$(x,y)\leftrightarrow(-y,-x)$, namely
%
\begin{eqnarray}\label{eqThmKernelPart}
\BK_m(x,y)&=&\frac{V_m}{(2\pi\I)^2} \oint_{\Gamma_0}dz\oint
_{\Gamma
_{0,z}}dw \frac{e^{t(z-z^{-1})}}{e^{t(w-w^{-1})}}\frac
{w^{y-m-1}}{z^{x-m}}\nonumber\\
&&\hspace*{97pt}{}\times\frac{ H_{2m+1}(w) H_{2m+1}(z^{-1})}{z-w}\nonumber
\\
&&{} + \frac{V_m}{(2\pi\I)^2} \oint_{\Gamma_0}dw\oint_{\Gamma_{0,w}}dz
\frac{e^{t(w-w^{-1})}}{e^{t(z-z^{-1})}}\frac{w^{y+m}}{z^{x+m+1}}\\
&&\hspace*{111pt}{}\times\frac{H_{2m+1}(z)H_{2m+1}(w^{-1})}{w-z}\nonumber\\
&&{} +\frac{V_m}{2\pi\I}\oint_{\Gamma_0}dz\frac{1}{z^{x-y+1}}
H_{2m+1}(z^{-1})H_{2m+1}(z)\nonumber
\end{eqnarray}
and
%
\begin{eqnarray}
\label{eqThmKernelPartTilde}
\widetilde\BK_m(x,y)&=&
-\frac{V_m}{(2\pi\I)^2} \oint_{\Gamma_0}dz\oint_{\Gamma_{0,z}}dw
\frac{e^{t(z-z^{-1})}}{e^{t(w-w^{-1})}}\frac{w^{y-m-1}}{z^{x-m}}\nonumber\\
&&\hspace*{105pt}{}\times\frac
{H_{2m+1}(w) H_{2m+1}(z^{-1})}{z-w}\nonumber\\
&&{} - \frac{V_m}{(2\pi\I)^2} \oint_{\Gamma_0}dw\oint_{\Gamma_{0,w}}dz
\frac{e^{t(w-w^{-1})}}{e^{t(z-z^{-1})}}
\frac{w^{y+m}}{z^{x+m+1}}\\
&&\hspace*{112pt}{}\times\frac{H_{2m+1}(z)H_{2m+1}(w^{-1})}{w-z}\nonumber\\
&&{} -\Id_{[x\neq y]}\frac{V_m}{2\pi\I}\oint_{\Gamma_0}dz\frac{1}{z^{x-y+1}}
H_{2m+1}(z^{-1})H_{2m+1}(z),\nonumber
\end{eqnarray}
where $V_m=1/(H_{2m+1}(0)H_{2m+2}(0))$. The function $H_n$ itself is a
Fredholm determinant on $\ell^2(\{n,n+1,\ldots\})$
%
\begin{equation}\label{eqHn}
H_n(z^{-1}):=\det\bigl(\Id-K(z^{-1})\bigr)_{\ell^2(\{n,n+1,\ldots\})}
\end{equation}
of the kernel
%
\begin{equation}\label{KernelInZvariable}\quad
K(z^{-1})_{k,\ell}:=\frac{(-1)^{k+\ell}}{(2\pi\I)^2}\oint_{\Gamma
_0}du \oint_{\Gamma_{0,u}}dv\frac{u^{\ell}}{v^{k+1}} \frac
{1}{v-u}\frac
{u-z}{v-z}\frac{e^{2t(u-u^{-1})}}{e^{2t(v-v^{-1})}},
\end{equation}
where $\Gamma_{0}$ is any anticlockwise simple loop enclosing $0$ and
similarly $\Gamma_{0,u}$ encircles~$0$ and $u$ only (hence not $z$).
\end{theorem}
\begin{pf}
\textit{Step} 1: \textit{Computing the kernel $\BK_m(x,y)$ for the inliers
$\mathbf{x}(\tau)$ at $\tau=0$, from the Karlin--McGregor
formula} (\ref{1}):
It is well known by~\cite{Bor98} that a measure of the form (\ref{1})
implies that the point process (random point measure)
$ \eta$, as in~(\ref{eta}),
is determinantal with correlation kernel
%
\begin{equation}\label{kernel1}
\BK_m(x,y)=\sum^{2m+1}_{k,\ell=1}\varphi_k(y)[A^{-1}]_{k,\ell}
\varphi
_{\ell}(x),\qquad x,y\in\Z,
\end{equation}
where
%
\begin{equation}\label{eqPhik}
\varphi_k(x)=I_{x+k-1-m}(2t),
\end{equation}
and $A$ is the $(2m+1)\times(2m+1)$ matrix with entries
%
\begin{equation}\label{eqMatrixA}
[A]_{k,\ell}\equiv\la\varphi_k,\varphi_{\ell}\ra=\sum_{x\in\Z
}\varphi
_k(x)\varphi_{\ell}(x).
\end{equation}
Using (\ref{eqbesselI}) and (\ref{eqPhik}), the entries of the
$(2m+1)\times(2m+1)$ matrix $A$, as in (\ref{eqMatrixA}), are given by
%
\begin{eqnarray}
A_{k,\ell}&=&\sum_{x\in\Z}\varphi_k(x)\varphi_\ell(x) =\sum
_{x\geq
0}\varphi_k(x)\varphi_\ell(x) + \sum_{x<0}\varphi_k(x)\varphi
_\ell
(x)\nonumber\\
&=&\sum_{x\geq0} \frac{1}{(2\pi\I)^2}\oint_{\Gamma_0}dz\oint
_{\Gamma
_0}dw \frac{e^{t(z+z^{-1})}e^{t(w+w^{-1})}}{z^{k} w^{\ell}}\frac
{1}{(zw)^{x-m}}\\
&&{}+\sum_{x<0} \frac{1}{(2\pi\I)^2}\oint_{\Gamma_0}dz\oint_{\Gamma_0}dw
\frac{e^{t(z+z^{-1})}e^{t(w+w^{-1})}}{z^{k}
w^{\ell}}\frac{1}{(zw)^{x-m}}.\nonumber
\end{eqnarray}
In the first integrals, we deform the paths to $|z|=1$ and $|w|=R>1$.
Then we take the sum inside the integrals and use $\sum_{x\geq
0}(zw)^{-x}=wz/(wz-1)$. Similarly, in the second integrals, we deform
the paths as $|z|=1$ and $|w|=1/R<1$ and use $\sum
_{x<0}(zw)^{-x}=-wz/(wz-1)$. This leads to
%
\begin{eqnarray}\label{319}
A_{k,\ell}&=&\frac{1}{(2\pi\I)^2}\oint_{|z|=1}dz\oint_{|w|=R}dw
\frac
{e^{t(z+z^{-1})}e^{t(w+w^{-1})}}{z^{k-m} w^{\ell-m}}\frac
{wz}{wz-1}\nonumber\\
&&{}-\frac{1}{(2\pi\I)^2}\oint_{|z|=1}dz\oint_{|w|=1/R}dw \frac
{e^{t(z+z^{-1})}e^{t(w+w^{-1})}}{z^{k-m} w^{\ell-m}}\frac{wz}{wz-1}\\
&=&\frac{1}{2\pi\I}\oint_{|z|=1}dz
\frac{e^{2t(z+z^{-1})}}{z^{k-\ell+1}}=I_{k-\ell}(4t),\nonumber
\end{eqnarray}
since for any value of $z$, the two integrals differ only by the
residue\footnote{This residue argument will reappear later in (\ref
{interchange}).}
at $w=1/z$. However, doing the asymptotics of the kernel $\BK_m(x,y)$
with this choice of basis and thus with this $A^{-1}$ seems to be hopeless.

\textit{Step} 2: \textit{Changing the basis $\varphi_k\mapsto\psi_k$,
such that $A\mapsto\Id$ in the kernel\break $\BK_m(x,y)$, that is, so that
$\BK_m(x,y)=\sum_{k=1}^{2m+1} \psi_k(x)\psi_k(y)$.}
Replace the basis  $(\varphi_k(x))_{k=1,\ldots,2m+1}$ with an
orthonormal basis $(\psi_k(x))_{k=1,\ldots,2m+1}$ with respect to the
$\ell^2(\Z)$ scalar product $\la{\,,\,}\ra$ used in (\ref{eqMatrixA})
[generating the same vector space, i.e., $\det(\varphi_k(x_j))_{1\leq
k,j\leq n}= \mathrm{const}\times\det(\psi_k(x_j))_{1\leq k,j\leq n}$ so
that the measure (\ref{1}) has the same form, but with $A=\Id$]. More
precisely, we shall search for polynomials $P_k$ of degree $k$ such
that, upon defining $d\rho_t(z):=\frac{dz}{2\pi iz}e^{t(z+z^{-1})}$,
%
\begin{eqnarray}\label{psi}
\psi_k(x)&=&\oint_{S^1}\frac{d\rho
_t(z)}{z^{x-m}}P_{k-1}(z^{-1})\nonumber\\[-8pt]\\[-8pt]
&=&\oint
_{S^1}d\rho_t(w)w^{x-m}P_{k-1}(w),\qquad 1\leq k\leq2m+1,\nonumber
\end{eqnarray}
satisfies, using the same argument as in (\ref{319}),
%
\begin{eqnarray}\label{in-prod2}\quad
\delta_{k,l} &=& \la\psi_k,\psi_{\ell}\ra\nonumber\\
&=&\sum_{x\in\Z}\oint
_{\Gamma
_0}d\rho_t(z)\oint_{\Gamma_0}d\rho_t(w) (zw)^{x-m}
P_{k-1}(z)P_{\ell
-1}(w)\nonumber\\[-8pt]\\[-8pt]
&=& \oint_{S^1}d\rho_{2t}(z) P_{k-1}(z)P_{\ell-1}(z^{-1})\nonumber\\
&=:&\lla
P_{k-1}, P_{\ell-1} \rra,
\nonumber
\end{eqnarray}
thus defining a new inner-product $ \lla{\,,\,}\rra$ on the
circle $S^1=\{z\in\C | |z|=1\}$. So it suffices to find an
orthonormal\vadjust{\goodbreak} basis of polynomials on the circle for the weight $d\rho
_{2t}(z)$. A classical expression for the polynomial $P_k(z)$ is (see,
e.g.,~\cite{Sze67})
%
\begin{equation}
P_k(z)=\frac{1}{\sqrt{\det m_k\cdot\det m_{k+1}}}
\det\pmatrix{
&1\cr
[ \mu_{i,j} ]\hspace*{-4pt}\displaystyle \mathop{\mathop{\hphantom{0}}_{0\leq i\leq k}}_{0\leq j
\leq k-1}
&z\cr
& \vdots\vspace*{2pt}\cr
&z^k},
\end{equation}
where $m_k=[\mu_{i,j}]_{0\leq i,j\leq k-1}$ and
%
\begin{equation}
\mu_{i,j}:= \lla z^i, z^j \rra=\oint_{S^1}d\rho
_{2t}(z)z^{i-j}=I_{i-j}(4t).
\end{equation}
Hence the $P_k(z)$ are polynomials of $z$ with real coefficients.
Orthonormal polynomials on the circle satisfy a
Christoffel--Darboux-type formula, due to Szeg\H{o}; see~\cite{Sim04}.
Namely, with the notation $P^*_n(z)=z^n \overline{P(\bar
z^{-1})}$ and further using the reality of the coefficients, one
obtains for $z,w\in S^1$,
%
\begin{eqnarray}\label{Darboux}
\sum^{n-1}_{\ell=0} {P_{\ell}(z^{-1})}P_{\ell}(w)&=&
\sum^{n-1}_{\ell=0}\overline{P_{\ell}(z)}P_{\ell}(w)\nonumber\\
&=&\frac
{\overline
{P^*_n(z)}P^*_n(w)-\overline{P_n(z)}P_n(w)}{1-\bar z
w}\nonumber\\[-8pt]\\[-8pt]
&=&\frac{\overline{z^n\overline{P_n(\bar z^{-1})}}w^n\overline
{P_n(\bar
w^{-1})}-\overline{P_n(z)}P_n(w)}{1-w/z}\nonumber\\
&=&\frac{z^{-n}P_n(z)w^nP_n(w^{-1})-P_n(z^{-1})P_n(w)}{1-w/z}.
\nonumber
\end{eqnarray}

\textit{Step} 3: \textit{Expressing the polynomials $P_n(z)$ in terms of
the Fredholm determinant $H_n(z^{-1})$}, \textit{as in} (\ref{eqHn}). In order
to do this, one first introduces the bilinear form
%
\begin{equation}\label{in-prod3}
\la f,g\ra_{{\mathbf t},{\mathbf s}}:=\frac{1}{2\pi\I}\oint_{S^1}\frac
{du}{u}f(u)g(u^{-1})e^{\sum_{j=1}^{\infty}(t_ju^j-s_ju^{-j})},
\end{equation}
upon setting ${\mathbf t}:=(t_1,t_2,\ldots)\in\C^{\infty}$ and ${\mathbf s}:=(s_1,s_2,\ldots)\in\C^{\infty}$.
It was shown in \mbox{\cite{AvM97,AvM00}} (see also the lecture notes \cite
{vM10}) that the functions\footnote{For $\alpha\in\C$, one defines
$[\alpha]=(\alpha,\frac{\alpha^2}{2},\frac{\alpha
^3}{3},\ldots
)\in\C^{\infty}$.}
%
\begin{eqnarray}\label{tau1}
p_n^{(1)}({\mathbf t},{\mathbf s};z)&:=&z^n\frac{\tau_n({\mathbf t}-[z^{-1}],{\mathbf s})}{\sqrt{\tau_n({\mathbf t},{\mathbf s})\tau_{n+1}({\mathbf t},{\mathbf
s})}},
\nonumber\\[-8pt]\\[-8pt]
p_n^{(2)}({\mathbf t},{\mathbf s};z)&:=&z^n\frac{\tau_n({\mathbf t},{\mathbf s}+[z^{-1}])}{\sqrt{\tau_n({\mathbf t},{\mathbf s})\tau_{n+1}({\mathbf t},{\mathbf s})}}
\nonumber
\end{eqnarray}
are bi-orthonormal polynomials with regard to the bilinear form (\ref
{in-prod3}). In the formulas above, the $\tau_n({\mathbf t},{\mathbf s})$ are
2-Toda $\tau$-functions and are defined as Toeplitz determinants, which
are also expressible as a Fredholm determinant of the kernel (\ref
{BO-kernel}) below, using the Borodin--Okounkov identity~\cite{BO99}.
We obtain
%
\begin{eqnarray}\label{327}\hspace*{5pt}
\tau_n({\mathbf t},{\mathbf s})&:=&\det\biggl[\frac{1}{2\pi\I}\oint
_{S^1}\frac
{du}{u} u^{k-\ell}
e^{\sum_{j=1}^\infty(t_ju^j-s_ju^{-j})}\biggr]_{1\leq k,\ell\leq
n}\nonumber\\[-8pt]\\[-8pt]
&=&Z({\mathbf t},{\mathbf s})\det\bigl(\Id-\mathbf{K}({\mathbf t},{\mathbf s})
\bigr)_{\ell^2(\{n,n+1,\ldots\})},\qquad
Z({\mathbf t},{\mathbf s}) :=e^{-\sum_{j=1}^{\infty}j t_js_j},\hspace*{-22pt}
\nonumber
\end{eqnarray}
where the kernel $\mathbf{K}({\mathbf t},{\mathbf s})$ is given by
%
\begin{equation} \label{BO-kernel}\quad
\mathbf{K}({\mathbf t},{\mathbf s})_{k,\ell}:=\frac{1}{(2\pi\I)^2} \oint
_{\Gamma_0}du\oint_{\Gamma_{0,u}}dv\frac{u^{\ell}}{v^{k+1}}\frac
{1}{v-u} \frac{e^{\sum_{j=1}^{\infty}(t_jv^{-j}+s_jv^j)}}{e^{\sum
_{j=1}^{\infty}(t_ju^{-j}+s_ju^j)}}.
\end{equation}
The coefficients $t_j, s_j$ have to be such that the expression $\sum
_{j=1}^\infty(t_ju^j-s_ju^{-j})$ appearing in the exponent of (\ref
{327}) is analytic in the annulus $\rho<|z|<\rho^{-1}$ for $0<\rho<1$.
Then, the Borodin--Okounkov identity (\ref{327}) gives a kernel
$\mathbf{K}({\mathbf t},{\mathbf s})$, with contours given by
$|u|=|v|^{-1}=\rho
'$, with $0<\rho<\rho'<1$. Assume, using Cauchy's theorem, that the
contours may be deformed to any circle of radius $0<\rho<1$.
Then, using $\sum_{j=1}^\infty(v/z)^j/j=-\ln(1-v/z)$ (for $|v/z|<1$),
we obtain
%
\begin{eqnarray}
&&
\mathbf{K}({\mathbf t},{\mathbf s}+[z^{-1}])_{k,\ell}\nonumber\\[-8pt]\\[-8pt]
&&\qquad=\frac{1}{(2\pi\I)^2}
\oint_{\Gamma_0}du\oint_{\Gamma_{0,u}}dv\frac{u^{\ell
}}{v^{k+1}}\frac
{1}{v-u}\frac{1-{u}/{z}}{1-{v}/{z}} \frac{e^{\sum
_{j=1}^{\infty
}(t_jv^{-j}+s_jv^j)}}{e^{\sum_{j=1}^{\infty}(t_ju^{-j}+s_ju^j)}}\nonumber
\end{eqnarray}
and
%
\begin{equation}
Z({\mathbf t},{\mathbf s}+[z^{-1}])=e^{-\sum_{j=1}^{\infty}j t_j(s_j
+z^{-j}/j)}=Z({\mathbf t},{\mathbf s}) e^{-\sum_{j=1}^{\infty}t_jz^{-j}}.
\end{equation}

We now specialize all this to the locus
%
\begin{equation}\label{locus}
{\cal L} =\{
{\mathbf t}=(2t,0,0,\ldots),
{\mathbf s}=(-2t,0,0,\ldots)
\}.
\end{equation}
On this locus, one checks that
$Z({\mathbf t},{\mathbf s})|_{\cal L} =e^{4t^2}$,
that
$\mathbf{K}({\mathbf t},{\mathbf s})$ and its translation, restricted to the
locus ${\cal L}$, are closely related to the kernel $K(z^{-1})$ defined
in (\ref{KernelInZvariable})\footnote{With $A\stackrel{\mathrm{conj}}{=}B$
we mean that the two kernels $A$ and $B$ are conjugate kernels. In the
present case, the conjugation factor is $(-1)^{k-\ell}$. We remind the
reader that two conjugate kernels define the same determinantal point process.}
%
\begin{eqnarray}\label{Klocus}
\mathbf{K}({\mathbf t},{\mathbf s})|_{\cal
L}&\stackrel{\mathrm{conj}}{=}&
K(0),\nonumber\\[-8pt]\\[-8pt]
\mathbf{K}({\mathbf t},{\mathbf s}+[z^{-1}])|_{\cal L}&\stackrel{\mathrm{conj}}{=}& K(z^{-1}),
\nonumber
\end{eqnarray}
and that the restriction of $\tau_n({\mathbf t},{\mathbf s})$ to ${\cal L}$
leads to the Fredholm determinant $H_n(z^{-1})$ as defined in (\ref{eqHn}),
%
\begin{eqnarray}\label{tau2}\quad
\tau_n({\mathbf t},{\mathbf s})|_{\cal L}&=& H_n(0) Z({\mathbf t},{\mathbf s})
|_{\cal L}=e^{4t^2}H_n(0),\nonumber\\
\tau_n({\mathbf t},{\mathbf s}+[z^{-1}])|_{\cal L}&=& H_n(z^{-1}) e^{-2t/z}
Z({\mathbf t},{\mathbf s})|_{\cal L}\\
&=&H_n(z^{-1}) e^{4t^2-2t/z}.
\nonumber
\end{eqnarray}
Moreover, the bilinear form $\la f,g\ra_{{\mathbf t},{\mathbf s}}$ defined
in (\ref{in-prod3}) reduces to the inner-product $\lla f,g \rra$ defined in (\ref{in-prod2}),
%
\begin{equation}
\la f,g\ra_{{\mathbf t},{\mathbf s}}|_{{\cal L}}=\frac{1}{2\pi\I}\oint
_{S^1}\frac{du}{u}e^{2t(u+u^{-1})}f(u)g(u^{-1}) = \lla f,g \rra.
\end{equation}
It follows that the bi-orthogonal functions for $\la f(z),g(z)\ra
_{{\mathbf t},{\mathbf s}}$, restricted to the locus~${\cal L}$, coincide with the
orthonormal polynomials defined by (\ref{in-prod2}), which by (\ref
{tau1}), (\ref{tau2}) and (\ref{Klocus}) yields
%
\begin{equation}\label{eqPolyn}
P_n(z)=p_n^{(1)}({\mathbf t},{\mathbf s};z)|_{\cal L}=p_n^{(2)}({\mathbf t},{\mathbf s};z)|_{\cal L} = \frac{z^n e^{-2t/z} H_n(z^{-1})}{\sqrt
{H_n(0) H_{n+1}(0)}},
\end{equation}
where
%
\begin{equation}
H_n(z^{-1})=\det\bigl(\Id- K(z^{-1})\bigr)_{\ell^2(\{n,n+1,\ldots\})}
\end{equation}
with the kernel $K(z^{-1})$ as in (\ref{KernelInZvariable}); this
follows from (\ref{Klocus}).
The fact that the $p_n^{(1)}$ and $p_n^{(2)}$ are equal on the locus
${\cal L}$ is a consequence of the symmetry of the inner-product
$\lla{\,,\,} \rra$, as in (\ref{in-prod2}). However, one easily
verifies it with the above formulas. The equivalence of the Fredholm
determinant parts is evident only after the change of variable $v\to
1/\tilde u$ and $u\to1/\tilde v$. Then, the kernel obtained for
$p_n^{(1)}$ is the transpose of the one for $p_n^{(2)}$.

\textit{Step} 4: \textit{Expressing the kernel $\BK_m(x,y)$ as} (\ref
{eqThmKernelPart}). Using this new basis $\psi_k$, as in (\ref{psi}),
and using the Christoffel--Darboux formula (\ref{Darboux}),
the kernel $\BK_m(x,y)$ becomes, by Step 2 (recall that $n=2m+1$),
%
\begin{eqnarray}\label{eq236}\quad
\BK_m(x,y)&=&\sum_{k=1}^{n} \psi_k(x)\psi_k(y)\nonumber\\
&\stackrel{*}{=}&\oint
_{S^1}d\rho_t(z)\oint_{S^1}d\rho_t(w) \frac{w^{y-m}}{z^{x-m}} \sum
_{k=0}^{n-1}P_k(z^{-1})P_k(w)\nonumber\\[-8pt]\\[-8pt]
&=&\oint_{\Gamma_0} d\rho_t(z)\oint_{\Gamma
_{0,z}} d\rho_t(w) \frac{w^{y-m}}{z^{x-m-1}}\frac{1}{z-w}
\biggl(\biggl(\frac{w}{z}\biggr)^n
P_n(z)P_n(w^{-1})\nonumber\\
&&\hspace*{188pt}{}-P_n(z^{-1})P_n(w)\biggr).
\nonumber
\end{eqnarray}
Note that the $w$-integrand in the double integral $\stackrel{*}{=}$
has no pole at $w=z$, enabling one to deform the $w$-contour so as to
include $z\in S^1$; this has the advantage that the double integral of
the difference can be written as the difference of two double
integrals, each of them being finite.

Inserting (\ref{eqPolyn}) into (\ref{eq236}) and setting
$V_m=1/(H_{2m+1}(0)H_{2m+2}(0))$ we get
%
\begin{eqnarray}\quad
&&
\BK_m(x,y)\nonumber\\
&&\qquad=\frac{V_m}{(2\pi\I)^2} \oint_{\Gamma_0}dz\oint
_{\Gamma
_{0,z}}dw \frac{e^{t(z-z^{-1})}}{e^{t(w-w^{-1})}}\frac
{w^{y-m-1}}{z^{x-m}}\frac{ H_{2m+1}(w)
H_{2m+1}(z^{-1})}{z-w}\\
&&\qquad\quad{} -\frac{V_m}{(2\pi\I)^2} \oint_{\Gamma_0}dz\oint_{\Gamma_{0,z}}dw
\frac{e^{t(w-w^{-1})}}{e^{t(z-z^{-1})}}
\frac{w^{y+m}}{z^{x+m+1}}\frac{H_{2m+1}(z)H_{2m+1}(w^{-1})}{z-w}.
\nonumber
\end{eqnarray}
The expression in (\ref{eqThmKernelPart}) is finally obtained by
noticing that
%
\begin{eqnarray}\label{interchange}
&&\frac{1}{(2\pi\I)^2}\oint_{\Gamma_0}dz\oint_{\Gamma
_{0,z}} dw \frac{F(z,w)}{w-z}\nonumber\\[-8pt]\\[-8pt]
&&\qquad=
\frac{1}{(2\pi\I)^2}\oint_{\Gamma_0}dw\oint_{\Gamma
_{0,w}} dz \frac{F(z,w)}{w-z}+\oint_{\Gamma_0} \frac{dz}{2\pi\I}
F(z,z),\nonumber
\end{eqnarray}
proving formula (\ref{eqThmKernelPart}).

\textit{Step} 5: \textit{Expressing the dual kernel $\widetilde\BK_m(x,y)$
as} (\ref{eqThmKernelPartTilde}). First of all, by (\ref{eqPolyn}),
we have
%
\begin{equation}
H_n(z^{-1})=P_n(z) e^{2t/z} z^{-n} \sqrt{H_n(0) H_{n+1}(0)}.
\end{equation}
Thus (with $n=2m+1$), the last term of (\ref{eqThmKernelPart}) is
given by
%
\begin{eqnarray}\label{eq240}
&&
\frac{V_m}{2\pi\I}\oint_{\Gamma_0}\frac{dz}{z^{x-y+1}}
H_{2m+1}(z^{-1})H_{2m+1}(z) \nonumber\\[-8pt]\\[-8pt]
&&\qquad= \frac{1}{2\pi\I}\oint_{S^1} \frac
{dz}{z^{x-y+1}} e^{2t(z+z^{-1})} P_n(z)P_n(z^{-1}).\nonumber
\end{eqnarray}
In particular, at $x=y$ we have
%
\begin{equation}
(\ref{eq240})|_{x=y}= \lla P_n, P_n \rra=1
\end{equation}
and thus
%
\begin{eqnarray}
&&\frac{V_m}{2\pi\I}\oint_{\Gamma_0}\frac{dz}{z^{x-y+1}}
H_{2m+1}(z^{-1})H_{2m+1}(z)\nonumber\\[-8pt]\\[-8pt]
&&\qquad=\delta_{x,y}+(1-\delta_{x,y})\frac{V_m}{2\pi\I}\oint_{\Gamma
_0}dz\frac{1}{z^{x-y+1}} H_{2m+1}(z^{-1})H_{2m+1}(z).
\nonumber
\end{eqnarray}
So, $\widetilde\BK_m(x,y)=\delta_{x,y}-\BK_m(x,y)=\widetilde\BK
^{\mathrm{ext}}_m(0,x_1;0,x_2)$ of (\ref{eqThmKernelHolesIntroExt}), thus
establishing Theorem~\ref{thmKernelPart}. This also ends the proof of
Theorem~\ref{MainTheorem} for $t_1=t_2=0$.
\end{pf}

\section{Reshaping, motivation and Bessel representation}\label{shape}
In this section we first reshape the kernel (\ref
{eqThmKernelHolesIntroExt}) of Theorem~\ref{MainTheorem} for
$t_1=t_2=0$, to make it adequate for asymptotic analysis. Second, we
rewrite all the terms using Bessel functions and the Bessel kernel.
This will allow us to use known asymptotics for Bessel functions and
kernel, without the need for new asymptotic analysis.

\subsection{Reshaping}\label{shape1}

Note that the kernel $K(z^{-1})$, defined in (\ref{KernelInZvariable}),
with $|u|<|v|<|z|$, namely
%
\begin{equation}\label{KernelInZvariable1}
K(z^{-1})_{k,\ell}:=\frac{(-1)^{k+\ell}}{(2\pi\I)^2}\oint_{\Gamma
_0}du \oint_{\Gamma_{0,u}}dv\frac{u^{\ell}}{v^{k+1}} \frac
{1}{v-u}\frac
{u-z}{v-z}\frac{e^{2t(u-u^{-1})}}{e^{2t(v-v^{-1})}}
\end{equation}
is a rank-one perturbation
%
\begin{equation}
K(z^{-1})_{k,\ell}=K(0)_{k,\ell} + h_k(z^{-1}) g_\ell
\end{equation}
of the symmetric\footnote{As is seen by replacing $u\mapsto
1/u, v\mapsto1/v$.} kernel
%
\begin{equation}\label{KernelZInfinity}
K(0)_{k,\ell} =\frac{(-1)^{k+\ell}}{(2\pi\I)^2}\oint_{\Gamma_0}du
\oint_{\Gamma_{0,u}}dv\frac{u^{\ell}}{v^{k+1}} \frac{1}{v-u}\frac
{e^{2t(u-u^{-1})}}{e^{2t(v-v^{-1})}},
\end{equation}
upon using the identity
%
\begin{equation}
\frac{1}{v-u}\frac{u-z}{v-z}=\frac{1}{v-u}-\frac{1}{v-z},
\end{equation}
where (remember $|v|<|z|$ in the first integration below)
%
\begin{eqnarray} \label{eqFk}
h_k(z^{-1}) &=& \frac{-1}{2\pi\I}\oint_{\Gamma_{0}}\frac
{dv}{(-v)^{k+1}}\frac{e^{-2t(v-v^{-1})}}{v-z} \nonumber\\
&=& \frac{-1}{2\pi\I}\oint_{\Gamma_{0,z}}\frac
{dv}{(-v)^{k+1}}\frac
{e^{-2t(v-v^{-1})}}{v-z}+\frac{e^{-2t(z-z^{-1})}}{(-z)^{k+1}}
\\
&=:&\bar{h}_k(z^{-1})+ \frac
{{e^{-2t(z-z^{-1})}}}{(-z)^{k+1}}\nonumber
\end{eqnarray}
and
%
\begin{equation}\label{eqGl}
g_\ell=\frac{-1}{2\pi\I}\oint_{\Gamma_0} du (-u)^\ell e^{2t(u-u^{-1})}.
\end{equation}
In (\ref{eqFk}), one has replaced the integration about a small circle
around $0$ by an integration about a contour containing $z$ as well;
this is done in order to be able to expand, later on, $1/(v-z)$ in a
power series in $z/v$.\vadjust{\goodbreak} Therefore we can rewrite the Fredholm
determinant $H_n(z^{-1})$ of $K(z^{-1})$ as
%
\begin{equation}\label{eqHnB}
H_n(z^{-1})=H_n(0)\bigl(1-R_n(z^{-1})\bigr),
\end{equation}
where\footnote{For $a=(a_k)_{k\in\Z}$ and $b=(b_k)_{k\in\Z}$, the
inner-product $\langle a,b \rangle:=\sum_{k\in\Z} a_k b_k$.}
%
\begin{eqnarray}\label{RQ}
R_n(z^{-1})&:=&\langle Q,\chi_n h(z^{-1})\rangle,\nonumber\\[-8pt]\\[-8pt]
Q_k&:=&\bigl(\bigl(\Id-\chi_n
K(0)\chi_n\bigr)^{-1}\chi_n g\bigr)_k\nonumber
\end{eqnarray}
and $\chi_n(k)=\Id_{[k\geq n]}$; here the symmetry of $K(0)$ is being used.
Accordingly $R_n(z^{-1})=\la Q,\chi_nh(z^{-1})\ra$, as in (\ref{RQ}),
decomposes as (recall that $n=2m+1$)
%
\begin{equation}
R_n(z^{-1})=S_n(z^{-1})+ \frac{e^{-2t(z-z^{-1})}}{(-z)^n}T_n(z^{-1})
\end{equation}
with
%
\begin{equation}\label{Rbar}
S_n(z^{-1})=\langle Q,\chi_n \bar{h}(z^{-1})\rangle,\qquad
T_n(z^{-1})= \sum_{k\geq1} \frac{Q_{n+k-1}}{(-z)^{k}}.
\end{equation}
We set for $x\in\Z$,
%
\begin{eqnarray}\label{eqABCD}
A(x)&:=&\frac{-1}{2\pi\I}\oint_{\Gamma_0} dz \frac
{e^{t(z-z^{-1})}}{(-z)^{x-m}}\bigl(1-S_n(z^{-1})\bigr),\nonumber\\
B(x)&:=&\frac{-1}{2\pi\I}\oint_{\Gamma_0} dz \frac
{e^{-t(z-z^{-1})}}{(-z)^{x+m+1}}T_n(z^{-1}),\nonumber\\
C_1(x)&:=&\frac{-1}{2\pi\I}\oint_{\Gamma_0} dz \frac{T_n
(z^{-1})T_n(z)}{(-z)^{x+1}},\\
C_2(x)&:=&\Id_{[x\neq0]}\frac{-1}{2\pi\I}\oint_{\Gamma_0} dz
\frac
{R_n(z^{-1})+R_n(z)-R_n(z^{-1})R_n(z)}{(-z)^{x+1}},\nonumber\\
C(x)&:=& 2C_1(x)+C_2(x).\nonumber
\end{eqnarray}
Remark that $C_1(x)=C_1(-x)$ and $C_2(x)=C_2(-x)$. Also introduce
functions $E_i(z,w)$, which also depend on $n=2m+1$,
%
\begin{eqnarray}\label{eqE1234}\quad
E_1(z,w)&:=&\frac{e^{t(z-z^{-1})}}{e^{t(w-w^{-1})}}\biggl(\frac
{z}{w}\biggr)^m\bigl(1-S_n(z^{-1})\bigr)\bigl(1-S_n(w)\bigr),\nonumber\\
E_2(z,w)&:=&-\frac{e^{t(z-z^{-1})}}{e^{-t(w-w^{-1})}}(-z)^m
(-w)^{m+1}\bigl(1-S_n(z^{-1})\bigr)T_n(w),\nonumber\\[-8pt]\\[-8pt]
E_3(z,w)&:=&-\frac
{e^{-t(z-z^{-1})}}{e^{t(w-w^{-1})}}(-z)^{-m-1}(-w)^{-m}
T_n(z^{-1})\bigl(1-S_n(w)\bigr),\nonumber\\
E_4(z,w)&:=&
-\frac{e^{t(z-z^{-1})}}{e^{t(w-w^{-1})}}\biggl(\frac{z}{w}
\biggr)^{m}T_{n}(z)T_{n}(w^{-1}).\nonumber
\end{eqnarray}
With these notations, the following statement holds.
%
\begin{proposition}\label{propDecomp}
The kernel $\widetilde\BK_m(x,y)$ in (\ref{eqThmKernelPartTilde}) has
the following expression:
%
\begin{eqnarray}\label{eq35c}\qquad
&&(-1)^{x-y}\frac{H_{n+1}(0)}{H_{n}(0)}\widetilde\BK_m(x,y)\nonumber\\
&&\qquad=C(x-y)\\
&&\qquad\quad{}+ \frac{1}{(2\pi\I)^2}\oint_{\Gamma_0}dz \oint_{\Gamma_{0,z}} dw
\frac{\sum_{i=1}^4 E_i(z,w)}{z-w}\biggl(\frac{(-w)^{y-1}}{(-z)^x}
+\frac{(-z)^{y }}{(-w)^{x+1}}\biggr)
\nonumber
\end{eqnarray}
as well as the Airy kernel-like expression:
%
\begin{eqnarray}\label{eqpropDecomp}
&&(-1)^{x-y}\frac{H_{n+1}(0)}{H_n(0)}\widetilde\BK
_m(x,y)\nonumber\\
&&\qquad=C(x-y)\nonumber\\
&&\qquad\quad{}+\sum_{c\geq0} \bigl(
A(x-c)A(y-c)+A(-x-c)A(-y-c)\nonumber\\[-8pt]\\[-8pt]
&&\hspace*{62.4pt}{}-A(x-c)B(y-c)-A(-x-c)B(-y-c)\nonumber\\
&&\hspace*{62.4pt}{}-B(x-c)A(y-c)-B(-x-c)A(-y-c)\bigr)\nonumber\\
&&\qquad\quad{}-\sum_{c<0} \bigl(B(x-c)B(y-c)+B(-x-c)B(-y-c)\bigr).
\nonumber
\end{eqnarray}
\end{proposition}
\begin{pf} Let us first prove (\ref{eq35c}). Consider the kernel
$\widetilde\BK_m(x,y)$ as in (\ref{eqThmKernelPartTilde}); one uses
$H_n(z^{-1})=H_n(0)(1-R_n(z^{-1}))$, as in (\ref{eqHnB}), and one
renames the integration variables $(w,z)\to(z,w)$ in the second double
integral, enabling us to combine the two double integrals. Then, taking
into account the prefactor,
%
\begin{eqnarray}\label{eq35a}\quad
&&(-1)^{x-y}\frac{H_{n+1}(0)}{H_{n}(0)}\widetilde\BK_m(x,y)\nonumber\\
&&\qquad= \frac
{\Id
_{[x\neq y]}}{2\pi\I}\oint_{\Gamma_0} \frac{dz}{(-z)^{x-y+1}}
\bigl(1-R_n(z^{-1})\bigr)\bigl(1-R_n(z)\bigr)\nonumber\\[-8pt]\\[-8pt]
&&\qquad\quad{}+\frac{1}{(2\pi\I)^2}\oint_{\Gamma_0}dz \oint_{\Gamma_{0,z}} dw
\frac
{e^{t(z-z^{-1})}}{e^{t(w-w^{-1})}}\biggl(\frac{z}{w} \biggr)^m
\biggl(\frac{(-w)^{y -1}}{(-z)^{x }}
+\frac{(-z)^{y }}{(-w)^{x +1}}\biggr)\nonumber\\
&&\hspace*{110pt}\qquad\quad{}\times\frac{(1-R_n(z^{-1}))(1-R_n(w))}{z-w}.
\nonumber
\end{eqnarray}
That the single integral above equals $C_2$, defined in (\ref{eqABCD}),
follows from the fact that the $-1$ term can be deleted, since $\frac
{1}{2\pi\I} \oint_{\Gamma_0} dz \,z^{y-x-1}=\delta_{x,y}$ and
$\delta_{x,y}\Id_{x\neq y}=0$. Multiply out
$(1-R_n(z^{-1}))(1-R_n(w))$, use the expression (\ref{Rbar}) of $R_n$
and the functions $E_i$'s defined in (\ref{eqE1234}) with the
result
%
\begin{eqnarray}\label{eq35}\quad
\mbox{(\ref{eq35a})}&=&\frac{1}{(2\pi\I)^2}\oint_{\Gamma_0}dz \oint
_{\Gamma
_{0,z}} dw \frac{1}{z-w}\biggl({\frac{(-w)^{y-1}} {(-z)^x} } + {\frac
{(-z)^y}{(-w)^{x+1}}} \biggr)\nonumber\\
&&\hspace*{46pt}{}\times\biggl(E_1(z,w)+E_2(z,w)
+E_3(z,w)-\frac{w}{z} E_4(w,z)
\biggr)\\
&&{}+C_2(x-y).
\nonumber
\end{eqnarray}
The double integral, involving the last expression in brackets, is not
in a usable form, in view of the saddle point method and the topology
of the contours (see the discussion after the proof). Namely, the \textit{integrations have to be interchanged}, at the expense of a residue
term, as is given by the general formula (\ref{interchange}). So, using
this formula, and further renaming $z\leftrightarrow w$, the double
integral with $E_4$ becomes
%
\begin{eqnarray}\label{eq35b}
&&\frac{1}{(2\pi\I)^2} \oint_{\Gamma_{0 }}dz\oint_{\Gamma_{0 ,z}}dw
\frac{1}{z-w} \biggl({\frac{(-w)^{y-1}} {(-z)^x} } + {\frac
{(-z)^y}{(-w)^{x+1}}} \biggr) E_4(z,w)\nonumber\\[-8pt]\\[-8pt]
&&\qquad{}+2C_1(x-y),\nonumber
\end{eqnarray}
where $C_1(x)$ is defined in (\ref{eqABCD}). So, taking equation (\ref
{eq35}) and (\ref{eq35b}) into account, we find that formula (\ref
{eq35c}) for the kernel $\widetilde\BK_m(x,y)$ holds.

Next we prove (\ref{eqpropDecomp}). The first observation is that the
kernel (\ref{eq35c}) depends on $x$ and $y$ through the expression in
brackets only; the latter itself is invariant for the interchange
$(x,y)\mapsto(-y,-x)$. So it suffices to consider the double integral
associated with the first term $(-w)^{y-1}(-z)^{-x}$ only; the other
one is automatic.
Since the integration paths can be taken to satisfy $|z|<|w|$, in the
double integral of (\ref{eq35c}), one may use the series
%
\begin{equation}
\frac{1}{z-w}=\frac{1}{(-w)}\sum_{c\geq0} \biggl(\frac
{-z}{-w}
\biggr)^c\qquad \mbox{valid for }|w|>|z|,
\end{equation}
and one notices that for each of the $E_i$, the double integral
decouples into the product of two integrals over $\Gamma_0$:
%
\begin{eqnarray}
&&\frac{1}{(2\pi\I)^2} \oint_{\Gamma_0}dz \oint_{\Gamma_{0,z}} dw
\frac
{E_1(z,w)}{z-w} \frac{(-w)^{y-1}} {(-z)^x}\nonumber\\
&&\qquad =\sum_{c\geq0} \oint_{\Gamma_0}\frac{-dz}{2\pi\I} \frac
{e^{t(z-z^{-1})} }{(-z)^{x-m-c}} \bigl(1-S_n(z^{-1})\bigr)
\nonumber\\[-8pt]\\[-8pt]
&&\qquad\hspace*{24pt}{}\times\oint_{\Gamma_0}\frac
{-dw}{2\pi\I}\frac{(-w)^{y-m-c-2}}{e^{t(w-w^{-1})}} \bigl(1-S_n(w)\bigr)\nonumber\\
&&\qquad=\sum_{c\geq0} A(x-c)A(y-c).\nonumber
\end{eqnarray}
To see that the second integral equals $A(y-c)$, one performs the
change of variable $w\mapsto1/w$. Since the only poles are at $w=0$
and $w^{-1}=0$, this is allowed; so, we do not pick up further poles.
The same decoupling occurs for the other $E_i$'s, which yields
%
\begin{eqnarray}\qquad
\frac{1}{(2\pi\I)^2}\oint_{\Gamma_0}dz \oint_{\Gamma_{0,z}} dw
\frac
{E_2(z,w)}{z-w} \frac{(-w)^{y-1}} {(-z)^x}&=& -\sum_{c\geq0}
A(x-c)B(y-c),\nonumber\\[-8pt]\\[-8pt]
\frac{1}{(2\pi\I)^2}\oint_{\Gamma_0}dz \oint_{\Gamma_{0,z}} dw
\frac
{E_3(z,w)}{z-w} \frac{(-w)^{y-1}} {(-z)^x}&=& - \sum_{c\geq0} B(x-c)A(y-c)
\nonumber
\end{eqnarray}
and
%
\begin{eqnarray}
&&
\frac{1}{(2\pi\I)^2}\oint_{\Gamma_0}dz \oint_{\Gamma_{0,z}} dw
\frac
{E_4(z,w)}{z-w} \frac{(-w)^{y-1}} {(-z)^{x}}\nonumber\\
&&\qquad=- \sum_{c\geq0} B(-x+c+1)B(-y+c+1) \\
&&\qquad=-\sum_{c< 0} B(-x-c )B(-y - c ).
\nonumber
\end{eqnarray}
Then adding the same expressions with the interchange $(x,y)\mapsto
(-y,-x)$ yields formula (\ref{eqpropDecomp}), completing the proof of
Proposition~\ref{propDecomp}.
\end{pf}

In anticipation of Section~\ref{IntegralReprKernel} on the integral
representation of the limiting kernel, which will be obtained by saddle
point analysis, some comments must be made here; they will also explain
the interchange of integrals, which occurred in (\ref{eq35b}). Given
the future rescaling $m \simeq2t$ with $x=\xi_1 t^{1/3}$, $y=\xi_2
t^{1/3}$ for $t\rightarrow\infty$, the steepest descent method applied
to $A(x)$ and $B(x)$ at $z=-1$, in particular to the part of the
integrand $ {e^{\pm t(z-z^{-1})}}{(-z)^{\pm m}}=e^{\pm tF(z)}$,
respectively, uses the Taylor expansions
%
\begin{eqnarray} \label{Taylor}
F(z)&:=& z-z^{-1}+2 \log(-z) = \tfrac13
(z+1)^3+\Or(z+1)^4,\nonumber\\[-8pt]\\[-8pt]
\log(-z)&=&-(z+1)-\tfrac12 (z+1)^2+\Or(z+1)^3.
\nonumber
\end{eqnarray}
The steepest descent path for $A(x)$ will therefore look like $\lcont$
with an angle of approximately $\pm\pi/3$, whereas for $B(x)$ it will
look like $\rcont$ with an angle of approximately\footnote{The angles
can be within the range $\pi/3 \pm\pi/6$ and $2\pi/3 \pm\pi/6$.}
$\pm
2\pi/3$ with the positive real axis. The contours of \textit{the four
double integrals} of equation (\ref{eq35c}), associated with each one
of the $E_i$'s, from the point of view of steepest descent analysis
about $z,w=-1$, are topologically two circles, a $z$-circle inside a
$w$-circle, which are deformed so that locally near $z=w=-1$ they look
like the set of pictures in Figure~\ref{figPaths1} (see
Section~\ref{IntegralReprKernel}), with the two circles intersecting
the real axis at the common point $z,w=-1$ and to the right of $-1$.

\subsection{Bessel reformulation}

The purpose of this section is to express the functions
$A(x), B(x), C_1(x)$ and $C_2(x)$, as in (\ref{eqABCD}) in terms of
Bessel functions, the expressions $Q_k$ and the Bessel kernel $K(0)$,
as in (\ref{RQ}) and (\ref{KernelZInfinity}). Throughout we will be
using the integral representation of the Bessel function of order $n\in
\Z$, together with its symmetries,
%
\begin{equation}\label{eqBesselDef}
J_n(2t)=\frac1{2\pi\I}\oint_{\Gamma_0}dz\frac
{e^{t(z-z^{-1})}}{z^{n+1}}=(-1)^n J_{-n}(2t)=(-1)^n J_n(-2t).
\end{equation}
$J_n(2t)$ is different from the modified Bessel function $I_n(2t)$,
defined in (\ref{eqbesselI}). To do so, we shall need the following
Bessel function expressions for the basic building blocks.
%
\begin{lemma}\label{lemmaBessel}
The kernel $K(0)$ defined in (\ref{KernelZInfinity}), the expressions $
h_k$ and $g_\ell$ given in (\ref{eqFk}) and (\ref{eqGl}) and the
functions $T_n(z^{-1})$ and $S_n(z^{-1})$, given in (\ref{Rbar}), can
be expressed in terms of Bessel functions as follows:
%
\begin{eqnarray}\label{eqLemBessel}
K(0)_{k,\ell}&=&\sum_{a\geq0} J_{k+a+1}(4t) J_{\ell+a+1}(4t)\nonumber\\
&=:&
B_{2t}(k+1,\ell+1),\qquad g_\ell=J_{\ell+1}(4t),\nonumber\\
{h}_k(z^{-1})&=&-\sum_{a\geq0} (-z)^a J_{k+a+1}(4t)+\frac
{e^{-2t(z-z^{-1})}}{(-z)^{k+1}}\nonumber\\[-8pt]\\[-8pt]
&=&\bar h_k(z^{-1})+
\frac{e^{-2t(z-z^{-1})}}{(-z)^{k+1}},\nonumber\\
T_n(z^{-1})&=&\sum_{k\geq n} \frac{Q_k}{(-z)^{k-n+1}}
,\nonumber\\
S_n(z^{-1})&=&-\mathop{\mathop{\sum_{a\geq0}}}_{k\geq
n}(-z)^aQ_kJ_{k+a+1}(4t)
,\nonumber
\end{eqnarray}
where $B_t(i,j)$ is the Bessel kernel in~\cite{PS02}. Also,
%
\begin{equation}\label{eq415}
Q_k=\sum_{\ell\geq n} P_{k,\ell} J_{\ell+1}(4t) \qquad
\mbox{with
}P_{k,\ell}=\bigl((\Id-\chi_n K(0) \chi_n)^{-1}\bigr)_{k,\ell}.
\end{equation}
\end{lemma}
\begin{pf}
$\!\!\!$For\vspace*{1pt} $K(0)_{k,\ell}$ one uses in (\ref{KernelZInfinity}) the series
$1/(v-u)=v^{-1}\sum_{a\geq0}(u/v)^a$ for $|u|<|v|$ and then (\ref
{eqBesselDef}). The same geometric series is used for $h_k(z^{-1})$
in (\ref{eqFk}) but with $u$ replaced by $z$, from which formula (\ref
{eqLemBessel}) for $h_k(z^{-1})$ and the formula for $S_n$ by (\ref
{Rbar}) follow. Finally, one has $g_\ell=(-1)^{\ell-1} J_{-1-\ell
}(2t)=J_{\ell+1}(2t)$.
\end{pf}

The more intricate term is $C_2$ from (\ref{eqABCD}).
%
\begin{lemma}\label{lemLastTerm}
The expression $C_2(x)$, as in (\ref{eqABCD}), equals
%
\begin{equation}
C_2(x)=\Id_{[x\neq0]} C_2^*(x),
\end{equation}
where
%
\begin{eqnarray}\label{eqLemLastTerm}\qquad
C_2^*(x)&=&(-1)^{x}\frac{1}{2\pi\I}\oint_{\Gamma_0}dz\frac{1}{z^{x+1}}
\bigl(R_n(z^{-1})+R_n(z)-R_n(z^{-1})R_n(z)\bigr)\nonumber\\[-7pt]\\[-7pt]
&=&\sum_{k\geq n}Q_k
\biggl(\Id_{[x\neq0]} J_{k-|x|+1}(4t)
-Q_{k+|x|}+ \sum_{\ell\geq n} Q_\ell K(0)_{k,\ell-|x|}\biggr).
\nonumber
\end{eqnarray}
\end{lemma}
\begin{pf} One first notices that the integrand in (\ref
{eqLemLastTerm}) is invariant under the mapping $z\mapsto z^{-1}$.
Then, using formula (\ref{eqLemBessel}) for $R_n(z^{-1})=\break\sum
_{k\geq n}Q_k h_k(z^{-1})$,\vspace*{2pt} one breaks up the calculation as follows:

(a) \textit{Terms from $R_n(z^{-1})+R_n(z)$}. We have
%
\begin{equation}
R_n(z^{-1})+R_n(z) = \sum_{k\geq n} Q_k \bigl(h_k(z^{-1})+h_k(z)\bigr)
\end{equation}
and thus, by integration, one checks first for $x>0$, then for $x<0$
and for \mbox{$x=0$}, that, using the symmetry properties of the Bessel
functions [see~(\ref{eqBesselDef})],
%
\begin{eqnarray}
&&\frac{1}{2\pi\I}\oint_{\Gamma_0}dz \frac
{h_k(z^{-1})+h_k(z)}{z^{x+1}} \nonumber\\[1pt]
&&\qquad=
(-1)^x \bigl(\Id_{x>0}J_{k+1- x }(4t)+\Id_{x<0}J_{k+1+ x }(4t)\bigr) \\[1pt]
&&\qquad=(-1)^{x}\Id_{[x\neq0]} J_{k+1-|x|}(4t).
\nonumber
\end{eqnarray}
Substituting into the left-hand side of (\ref{eqLemLastTerm}) gives the
first term on the right-hand side of (\ref{eqLemLastTerm}).

(b) \textit{Terms from $R_n(z^{-1})R_n(z)$}. We have
%
\begin{equation}\label{eq232}
R_n(z^{-1})R_n(z) = \sum_{k,\ell\geq n} Q_k Q_\ell h_k(z^{-1})h_\ell(z).
\end{equation}
From (\ref{eqBesselDef}) and (\ref{eqLemBessel}) we get
%
\begin{eqnarray}
&&
\frac{1}{2\pi\I}\oint_{\Gamma_0}dz \frac{h_k(z^{-1})h_\ell
(z)}{z^{x+1}} \nonumber\\
&&\qquad= \sum_{a,b\geq0} (-1)^{a-b}\delta_{a-b,x}
J_{k+a+1}(4t) J_{b+\ell+1}(4t)\nonumber\\
&&\qquad\quad{}-\sum_{b\geq0} (-1)^{b+k+1} J_{\ell
+b+1}(4t)J_{k+b+1+x}(-4t)\nonumber\\[-8pt]\\[-8pt]
&&\qquad\quad{}-\sum_{a\geq0}(-1)^{a+\ell+1}J_{k+a+1}(4t)J_{x-\ell
-a-1}(4t)+(-1)^{x}\delta_{\ell-k,x}\nonumber\\
&&\qquad=(-1)^{x}\biggl(\delta_{\ell-k,x}-\sum_{a\geq0} J_{k+a+1}(4t)
J_{\ell
+a+1-|x|}(4t)\biggr)\nonumber\\
&&\qquad= (-1)^x\bigl(\delta_{\ell-k,x}-K(0)_{k,\ell-|x|} \bigr)
,\nonumber
\end{eqnarray}
using in the last equality the expression (\ref{eqLemBessel}) for the
kernel $K(0)$. In the second equality we used the symmetries (\ref
{eqBesselDef}) of the Bessel functions. Substituted into~(\ref
{eq232}), this gives the last two terms in (\ref{eqLemLastTerm}).
\end{pf}
%
\begin{proposition}\label{LemmaABCD}
The expressions $A(x), B(x), C(x)$, defined in (\ref{eqABCD}) for
$x\in
\Z$, can be expressed in terms of Bessel functions $J_k$, $Q_k$ and the
kernel $K(0)$, as follows:
%
\begin{eqnarray}\label{eq314}
A(x) &=& J_{m+1-x}(2t)+\sum_{k\geq n}\sum_{a\geq0} Q_k
J_{k+1+a}(4t)J_{m+1+a-x}(2t),\nonumber\\[-8pt]\\[-8pt]
B(x) &=& \sum_{k\geq n} Q_k J_{k-m+x}(2t)\nonumber
\end{eqnarray}
and
%
\begin{eqnarray}\label{eq315}
C(x)&=&\sum_{k\geq n} Q_k \bigl(J_{k-x+1}(4t)+J_{k+x+1}(4t)
\bigr)\nonumber\\[-8pt]\\[-8pt]
&&{}+\sum
_{k,\ell\geq n} Q_k Q_\ell\bigl(K(0)_{k+x,\ell} +K(0)_{k-x,\ell
}\bigr).\nonumber
\end{eqnarray}
\end{proposition}
\begin{pf}
The formulas for $A$ and $B$ follow directly from (\ref{eqABCD}) and
the expressions for $T_n$ and $S_n$ in (\ref{eqLemBessel}), together
with the symmetries (\ref{eqBesselDef}) of the Bessel functions. Then
%
\begin{eqnarray}
C_1(x)&=&\frac{-1}{2\pi\I}\oint_{\Gamma_0} dz \frac
{T_n(z^{-1})T_n(z)}{(-z)^{x+1}}\nonumber\\
&=&\sum_{k,\ell\geq n}Q_kQ_\ell\frac{(-1)^x}{2\pi\I}\oint_{\Gamma
_0}dz \frac{(-z)^{k-\ell}}{z^{x+1}}\\
&=&
\sum_{k,\ell\geq n}Q_k Q_\ell\delta_{k-\ell,x}=\sum_{k\geq n} Q_k
Q_{k+|x|}.
\nonumber
\end{eqnarray}
From Lemma~\ref{lemLastTerm}, it follows that
%
\begin{equation}\label{eq315b}\qquad
C_2(x)=\Id_{[x\neq0]}\sum_{k\geq n} Q_k
\biggl(J_{k-|x|+1}(4t)-Q_{k+|x|}+\sum_{\ell\geq n} Q_\ell K(0)_{k,\ell-|x|}
\biggr).
\end{equation}
Next we show that $\Id_{[x\neq0]}$ can actually be omitted. To do so,
it suffices to show that the sum on the right-hand side of (\ref
{eq315b}) vanishes when $x=0$.

Indeed, setting $P=(\Id-\chi_n K(0)\chi_n)^{-1}$, as in (\ref{eq415}),
remember that $g_\ell=J_{\ell+1}(4t)$ and that $Q_k=(P\chi_n g)_k$.
Then, denoting $\langle\cdot, \cdot\rangle$ the canonical scalar
product on $\ell^2(\Z)$ we get, for $x=0$, that the right-hand side of (\ref
{eq315b}) equals
%
\begin{eqnarray}
&&\langle P \chi_n g, \chi_n g\rangle-\langle P \chi_n g,\chi_n
P\chi_n
g\rangle+\langle P \chi_n g,\chi_n K(0)\chi_n P\chi_n g\rangle
\nonumber
\\
&&\qquad=\langle P \chi_n g, \chi_n g\rangle-\bigl\langle P \chi_n g,\chi_n
\bigl(\Id
-\chi_n K(0)\chi_n \bigr) P\chi_n g\bigr\rangle\\
&&\qquad=\langle P \chi_n g, \chi_n g\rangle- \langle P \chi_n g,\chi_n
g\rangle
=0.\nonumber
\end{eqnarray}
Plugging these results into $C(x)=2C_1(x)+C_2(x)$ we obtain
%
\begin{equation}\label{eq428}
C(x)=\sum_{k\geq n} Q_k \biggl(J_{k-|x|+1}(4t)+Q_{k+|x|}+\sum_{\ell
\geq
n} Q_\ell K(0)_{k,\ell-|x|} \biggr).
\end{equation}
It follows from the relation $P=\Id+\chi_nK(0)\chi_nP$ [see the
definition of $Q$ and $P$ in (\ref{eq415})] that acting on $\chi_ng$
and taking the $k$th entry,
%
\begin{equation}
Q_k=\Id_{[k\geq n]} \biggl(J_{k+1}(4t)+\sum_{\ell\geq n}
K(0)_{k,\ell}
Q_\ell\biggr).
\end{equation}
Using this relation for $Q_{k+|x|}$ in (\ref{eq428}) we obtain
%
\begin{eqnarray}
C(x)&=&\sum_{k\geq
n}Q_k\bigl(J_{k-|x|+1}(4t)+J_{k+|x|+1}(4t)\bigr)\nonumber\\[-8pt]\\[-8pt]
&&{}+\sum_{k,\ell\geq n}Q_k Q_\ell\bigl(K(0)_{k,\ell
-|x|}+K(0)_{k+|x|,\ell
}\bigr).
\nonumber
\end{eqnarray}
Finally, since $K(0)$ is symmetric, we replace $K(0)_{k,\ell
-|x|}=K(0)_{\ell-|x|,k}$ and change the labeling $k\leftrightarrow
\ell
$. This yields (\ref{eq315}), except for replacing $|x|$ by $x$, which
can then be done.
\end{pf}

\section{Extended kernel for finite time}\label{s5}
Formula (\ref{eqpropDecomp}) (in Proposition~\ref{propDecomp}) with
$A(x),B(x),C(x)$ given by Proposition~\ref{LemmaABCD} gives the kernel
$\widetilde\BK_m^{\rm}$ governing the fluctuations of the walkers
near the point of meeting of the two groups of nonintersecting random
walkers at time $\tau=0$. In this section we prove Theorem \ref
{MainTheorem} and we extend Proposition~\ref{propDecomp} to the
multitime setting (Theorem~\ref{ThmExtKernel}).

Consider the $n=2m+1$ walks whose positions were denoted by $x_k(\tau)$
in Section~\ref{SectFiniteSyst}. Consider $p$ different time slices
$\tau_1<\tau_2<\cdots<\tau_p$ in the interval $(-t,t)$. Then, the
probability measure at these times of the positions of the random walks
is given by
%
\begin{eqnarray}\label{eq51}
&& \Pb\Biggl(\bigcap_{j=1}^p\bigcap_{k=1}^{n}\{x_k(\tau_j)=y_k^j\}
\bigg|
\bigcap_{k=1}^{n}\{x_k(t)=x_k(-t)=m+1-k\}\Biggr)\nonumber\\
&&\qquad=\mathrm{const}\times\det[p_{t+\tau_1}(m+1-i,y_j^1) ]_{1\leq
i,j\leq n}\nonumber\\[-8pt]\\[-8pt]
&&\qquad\quad{}\times\Biggl(\prod_{\ell=1}^{p-1}\det[p_{\tau_{\ell+1}-\tau
_\ell
}(y_i^\ell,y_j^{\ell+1}) ]_{1\leq i,j\leq n}\Biggr)\nonumber\\
&&\qquad\quad{}\times\det[p_{t-\tau_p}(y_i^p,m+1-j) ]_{1\leq i,j\leq
n}.\nonumber
\end{eqnarray}
It is well known that a measure of this form has determinantal
correlations in space--time~\cite{EM97,TW98,FN98,Jo03b,RB04}, as stated
in the following proposition.
%
\begin{theorem}\label{ThmDetCorrExtended} Any probability measure on
$\{
x_i^{(\ell)},1\leq i \leq n,1\leq\ell\leq p\}$ of the form\footnote
{The functions $\phi(\tau_\ell,x ;\tau_{\ell+1},y)$ themselves may in
fact vary with $\ell$ above.}
%
\begin{eqnarray}\label{eq31}\quad
&&\frac{1}{Z}\det\bigl(\phi\bigl(\tau_0,a_i;\tau_1,x^{(1)}_j\bigr)
\bigr)_{1\leq
i,j\leq n}
\prod_{\ell=1}^{p-1}\det\bigl(\phi\bigl(\tau_\ell,x^{(\ell)}_i;\tau
_{\ell
+1},x^{(\ell+1)}_j\bigr)\bigr)_{1\leq i,j\leq n}\nonumber\\[-8pt]\\[-8pt]
&&\qquad{}\times\det\bigl(\phi\bigl(\tau_p,x^{(p)}_i;\tau_{p+1},b_j\bigr)
\bigr)_{1\leq
i,j\leq n}
\nonumber
\end{eqnarray}
has, assuming $Z\neq0$, the following determinantal $k$-point
correlation functions for $t_1,\ldots,t_k\in\{\tau_1,\ldots,\tau
_p\}$:
%
\begin{equation}
\rho^{(k)}(t_1,x_1,\ldots,t_k,x_k)=\det
(K(t_i,x_i;t_j,x_j)
)_{1\leq i,j\leq k}.
\end{equation}
The space--time kernel $K$ (often called extended kernel) is given by
%
\begin{eqnarray}
\label{eqExtKernelGeneral}
K(t_1,x_1;t_2,x_2)&=&
-\phi(t_1,x_1;t_2,x_2)\Id(t_2>t_1)\nonumber\\[-8pt]\\[-8pt]
&&{}+ \sum_{i,j=1}^n\phi(t_1,x_1;\tau_{p+1},b_i) [B^{-1}]_{i,j} \phi
(\tau
_0,a_j;t_2,x_2)
\nonumber
\end{eqnarray}
with ($*$ means integration with regard to the consecutive dots)
%
\begin{equation}\label{35}
\phi(\tau_r,x;\tau_s,y)=\cases{
\phi(\tau_r,x;\tau_{r+1},\cdot)*\cdots*\phi(\tau_{s-1},\cdot
;\tau
_s,y), &\quad if $\tau_r<\tau_s$,\cr
0, &\quad if $\tau_r\geq\tau_s$,}\hspace*{-26pt}
\end{equation}
and with the $n\times n$ matrix $B$ having entries $B_{i,j}=\phi(\tau
_0,a_i;\tau_{p+1},b_j)$.
\end{theorem}

Our measure (\ref{eq51}) has the form required by Theorem \ref
{ThmDetCorrExtended}. The normalization constant $Z$ is nothing else
but the partition function and it is nonzero since the set of $n$
paths satisfying the nonintersection constraint is nonempty. We
already determined the one-time kernel for $\tau=0$. To get the
extended kernel one has to let the one-time kernel ``evolve'' by means
of the operator of the random walk. This formulation was already
present in the work of Pr\"ahofer and Spohn on the Airy$_2$
process~\cite{PS02}.
%
\begin{lemma} \label{LemmaExt} The extended kernel $\widetilde\BK
_m^{\mathrm{ext}}(t_1,x_1;t_2,x_2)$ of the time-dependent point process
$\tilde\eta(\tau,x)$ is given in terms of the kernel $\widetilde\BK
_m(x_1,x_2)=\widetilde\BK_m^{\mathrm{ext}}(0,x_1;\break 0,x_2)$ of the same point
process $\tilde\eta(x)$ at $\tau=0$ by the formula
%
\begin{eqnarray} \label{Kext0}
\widetilde\BK_m^{\mathrm{ext}}(t_1,x_1;t_2,x_2)
&=& -\Id_{[t_2<t_1]} \bigl(e^{(t_2-t_1){\cal
H}}\bigr)(x_1,x_2)\nonumber\\[-8pt]\\[-8pt]
&&{}+
(e^{-t_1 \cal H} \widetilde\BK_m e^{t_2 \cal H}
)(x_1,x_2) ,\nonumber
\end{eqnarray}
where the infinitesimal generator ${\cal H}$ of the single random walk,
the discrete Laplacian, acts on functions $f$ as
%
\begin{equation}\label{H}
{\cal H}f(x)=f(x+1)+f(x-1)-2f(x),\qquad x\in\Z.
\end{equation}
\end{lemma}

Comparing the first term of (\ref{eqExtKernelGeneral}) and (\ref
{Kext0}), one sees a different ordering in the times. This is
consequence of the dual transformation.
\begin{pf*}{Proof of Lemma~\ref{LemmaExt}}
The operator ${\cal H} $ in (\ref{H}) is the generator of the
continuous time process defined by the transition probability
$p_t(x,y)$, in (\ref{tr-pr0}). Indeed, one checks that this transition
probability is given by [the reader is reminded of the notation
following formula (\ref{Laplacian0})]
%
\begin{eqnarray}\label{tr-pr1}
p_t(x,y)&=& e^{-2t} I_{|x-y|}(2t)=\frac{1}{2\pi\I} \oint_{\Gamma
_0}dz
\frac{e^{t(z+z^{-1}-2)}}{z^{x-y+1}}\nonumber\\[-8pt]\\[-8pt]
&=&e^{t{\cal H}}\Id(x,y)=(e^{t{\cal
H}})(x,y),\nonumber
\end{eqnarray}
because
%
\begin{eqnarray}
\frac{\partial}{\partial t}p_t(x,y)
&=&\frac{1}{2\pi\I} \oint_{\Gamma_0}\frac
{dz}{z^{x-y+1}}(z+z^{-1}-2)e^{t(z+z^{-1}-2)}\nonumber\\
&=&p_t(x-1,y)+p_t(x+1,y)-2p_t(x,y)\\
&=&({\cal H} p_t)(x,y)
\nonumber
\end{eqnarray}
with initial conditions $p_0(x,y)=\Id(x,y)$. Here, $\Id$ denotes the
identity operator on~$\Z$, that is, $\Id(x,y)=1$ if $x=y$ and $\Id
(x,y)=0$ if $x\neq y$. The one-point kernel in Section \ref
{SectFiniteSyst}, formula (\ref{eq236}), was written as a sum
involving $\psi_k(x)$ and $\psi_k(y)$. Under the time flow, they will
become different functions; therefore, we set $\Psi_k(0,x)=\Phi
_k(0,x)=\psi_k(x)$, and thus, with this new notation, the kernel reads
%
\begin{equation}
\BK_m(x_1,x_2)=\sum_{k=1}^{n} \psi_k(x_1)\psi_k(x_2) =
\sum_{k=1}^n \Psi_k(0,x_1)\Phi_k(0,x_2).
\end{equation}
The two set of functions $\{\Phi_k(0,x),k=1,\ldots,n\}$ and $\{\Psi
_k(0,x),k=1,\ldots,n\}$ satisfy
%
\begin{eqnarray}\qquad
\operatorname{span}\{\Phi_k(0,x),k=1,\ldots,n\}&=&\operatorname{span}\{
p_t(m+1-k,x),k=1,\ldots,n\},\nonumber\\
\operatorname{span}\{\Psi_k(0,x),k=1,\ldots,n\}&=&\operatorname{span}\{
p_t(x,m+1-k),k=1,\ldots,n\}\\
&&\eqntext{\mbox{with }
\langle\Phi_k(0,x),\Psi_p(0,x)\rangle= \delta_{k,p},}
\end{eqnarray}
so that the matrix $B$ defined in (\ref{eqExtKernelGeneral}) becomes
the identity matrix.

Let us consider the functions of Theorem~\ref{ThmDetCorrExtended}.
First of all, the function $\phi{(t_1,x_1;t_2,x_2)}$ appearing
in (\ref
{eqExtKernelGeneral}) becomes
%
\begin{eqnarray}\label{517}
\phi{(t_1,x_1;t_2,x_2)}\Id_{[t_2>t_1]} &=& \Id
_{[t_2>t_1]}p_{t_2-t_1}(x_1,x_2)\nonumber\\[-8pt]\\[-8pt]
&=&\Id_{[t_2>t_1]}\bigl(e^{(t_2-t_1){\cal
H}}\bigr)(x_1,x_2),\nonumber
\end{eqnarray}
where $t_1,t_2\in\{\tau_1,\ldots,\tau_p\}$. Next, with $\tau_0=-t$,
$\tau_{p+1}=t$ we have
%
\begin{equation}
\phi(t_1,x;t,b_k)=p_{t-t_1}(x,b_k)=(e^{-t_1{\cal H}})(x,\cdot)*\phi
(0,\cdot;t,b_k)
\end{equation}
and
%
\begin{equation}
\phi(-t,a_k;t_2,x)=p_{t_2+t}(a_k,x)=\phi(-t,a_k;0,\cdot
)*(e^{t_2{\cal
H}})(\cdot,x).
\end{equation}
With the choice of basis used for the kernel at $\tau=0$, we have that
$\phi(0,\cdot;t,b_k)$ is replaced by $\Psi_k(0,\cdot)$ and $\phi
(-t,a_k;0,\cdot)$ by $\Phi_k(0,\cdot)$ (so that $B=\Id$). Thus in
Theorem~\ref{ThmDetCorrExtended} we have replaced
%
\begin{eqnarray}
\phi(t_1,x;t,b_k)&\to&(e^{-t_1{\cal H}})(x,\cdot)*\Psi_k(0,\cdot
)\nonumber\\[-8pt]\\[-8pt]
&=&(e^{-t_1{\cal H}}\Psi_k(0,\cdot))(x)=:\Psi_k(t_1,x)\nonumber
\end{eqnarray}
and
%
\begin{eqnarray}
\phi(-t,a_k;t_2,x)&\to&\Phi_k(0,\cdot)*(e^{t_2{\cal H}})(\cdot
,x)=(\Phi
_k(0,\cdot) e^{t_2{\cal H}})(x) \nonumber\\[-8pt]\\[-8pt]
&=& (e^{t_2{\cal H}^\top}\Phi
_k(0,\cdot))(x)=:\Phi_k(t_2,x).
\nonumber
\end{eqnarray}
Therefore the extended kernel has the following expression in terms of
the kernel $\BK_m$ in (\ref{eq236}):
%
\begin{eqnarray}\label{eqExtKernel}\qquad
\BK_m^{\mathrm{ext}}(t_1,x_1;t_2,x_2)
&=&- \Id_{[t_1<t_2]}p_{t_2-t_1}(x_1,x_2)
+\sum_{k=1}^n \Psi_k(t_1,x_1)\Phi_k(t_2,x_2)\nonumber\\[-8pt]\\[-8pt]
&=& - \Id_{[t_1<t_2]}\bigl(e^{(t_2-t_1){\cal H}}\bigr)(x_1,x_2)
+(e^{-t_1 \cal H}\BK_m e^{t_2 \cal H})(x_1,x_2).
\nonumber
\end{eqnarray}

Notice that, using the semi-group property of $e^{t{\cal H}}$, we have
the consistency relations (for $i=1,\ldots,p$)
%
\begin{eqnarray}
\Psi_k(\tau_i,x)&=&\bigl(e^{(\tau_p-\tau_i){\cal H}} \Psi_k(\tau
_p,\cdot
)\bigr)(x),\nonumber\\[-8pt]\\[-8pt]
\Phi_k(\tau_i,x)&=&\bigl(\Phi_k(\tau_1,\cdot)e^{(\tau_i-\tau
_1){\cal H}}
\bigr)(x).
\nonumber
\end{eqnarray}

The kernel $\widetilde\BK_m^{\mathrm{ext}}$ for the dual random walk is
then given by taking the complement. Using (\ref{eqExtKernel}) and
remembering that $\widetilde\BK_m=\Id-\BK_m$ from formula (\ref
{kernel1b}), we get
%
\begin{eqnarray} \label{519}
&&\widetilde\BK_m^{\mathrm{ext}}(t_1,x_1;t_2,x_2)\nonumber\\
&&\qquad=\Id_{[t_1=t_2]}\Id(x_1,x_2)-\BK_m^{\mathrm{ext}}(t_1,x_1;t_2,x_2)\nonumber
\\
&&\qquad= \Id_{[t_1=t_2]}\Id(x_1,x_2)+\Id_{[t_1<t_2]}\bigl(e^{(t_2-t_1){\cal
H}}\bigr)(x_1,x_2)\nonumber\\
&&\qquad\quad{}-(e^{-t_1{\cal H}}\BK_m e^{t_2{\cal H}}
)(x_1,x_2)\\
&&\qquad=\Id_{[t_1=t_2]}\bigl(e^{(t_2-t_1){\cal H}}\bigr)(x_1,x_2)+
\Id_{[t_1<t_2]}\bigl(e^{(t_2-t_1){\cal H}}\bigr)(x_1,x_2) \nonumber\\
&&\qquad\quad{}- \bigl(e^{(t_2-t_1){\cal H}}\bigr)(x_1,x_2)+\bigl(e^{-t_1{\cal H}}(\Id- \BK
_m)e^{t_2 {\cal H}}\bigr)(x_1,x_2) \nonumber\\
&&\qquad= -\Id_{[t_2<t_1]}\bigl(e^{(t_2-t_1){\cal H}}\bigr)(x_1,x_2)+ (e^{-t_1{\cal H}}
\widetilde\BK_m e^{t_2 {\cal
H}})(x_1,x_2),\nonumber
\end{eqnarray}
yielding (\ref{Kext0}), completing the proof of Lemma~\ref{LemmaExt}.
\end{pf*}

With the help of Lemma~\ref{LemmaExt}, we can easily prove
Theorem~\ref{MainTheorem}, starting from Theorem~\ref{thmKernelPart}.
\begin{pf*}{Proof of Theorem~\ref{MainTheorem}}
One of the key ingredients is that $f(x):=u^x$ is an eigenfunction of
$\cal H$ with eigenvalue $u+u^{-1}-2$. Indeed,
%
\begin{eqnarray}
({\cal H} f)(x)&=& u^{x+1}+u^{x-1}-2 u^x =
(u+u^{-1}-2)u^x\nonumber\\[-8pt]\\[-8pt]
&=&(u+u^{-1}-2) f(x).\nonumber
\end{eqnarray}
Moreover, $\cal H$ is symmetric. Therefore,
%
\begin{eqnarray}\label{S1}
(e^{t{\cal H}} f)(x) &=& e^{t(u+u^{-1}-2)}f(x),\nonumber\\[-8pt]\\[-8pt]
(f e^{t{\cal H}})(x) &=& (e^{t{\cal
H}^\top}f)(x)=e^{t(u+u^{-1}-2)}f(x).\nonumber
\end{eqnarray}
Then, (\ref{eqThmKernelHolesIntroExt}) follows straightforwardly from
(\ref{eqThmKernelPartTilde}) by applying $e^{-t_1 {\cal H}}$ to the
left, $e^{t_2 {\cal H}}$ to the right of $\widetilde\BK_m$ [together
with (\ref{tr-pr1}) for the first term of (\ref{Kext0})].
\end{pf*}

For the further analysis, we extend the reformulation of the kernel for
$\tau=0$, as in Proposition~\ref{propDecomp}, to the extended case. For
that purpose, we first define the basic functions replacing $A$, $B$,
and $C$ of the one-time case (see Proposition~\ref{LemmaABCD}). To do
so, define a new function $J_x^{(\tau)}(2t)$ dependent on a parameter
$\tau$
%
\begin{eqnarray}\label{eq613}
J_{x}^{^{(\tau)}}(2t)&:=&\oint_{\Gamma_0} \frac{dz}{2\pi\I z} \frac
{e^{t(z-z^{-1})}}{z^x} e^{\tau(z+z^{-1}-2)}\nonumber\\[-8pt]\\[-8pt]
&=&e^{-2\tau}\biggl(\frac
{t+\tau
}{t-\tau}\biggr)^{x/2} J_{x}\bigl(2\sqrt{t^2-\tau^2}\bigr).\nonumber
\end{eqnarray}
Also define a $\tau$-dependent extension of the kernel
$K^{}(0)_{k,\ell
}$, as in (\ref{eqLemBessel}), namely
%
\begin{equation}\label{55}
K^{(\tau)}(0)_{k,\ell} := \sum_{a\geq0} J^{(\tau
)}_{a+k+1}(4t)J_{a+\ell+1}(4t).
\end{equation}
Then define new functions $A(\tau,x),B(\tau,x),C(\tau,x)$ with $\tau
\in\R$ and $x\in\Z$, which extend the functions $A(x),B(x),C(x)$,
first defined in (\ref{eqABCD}) and re-expressed in (\ref{eq314}), by
%
\begin{eqnarray}\label{53}
A(\tau,x)&:=&J^{(\tau)}_{m+1-x}(2t)+\sum_{k\geq n}\sum_{a\geq0} Q_k
J_{k+1+a}(4t) J^{(\tau)}_{m+1+a-x}(2t),\nonumber\\
B(\tau,x)&:=&\sum_{k\geq n} Q_k
J^{(\tau)}_{k-m+x}(2t),\nonumber\\[-8pt]\\[-8pt]
C(\tau,x)&:=&\sum_{k\geq n} Q_k \bigl(J^{(\tau
)}_{k+1+x}(4t)+J^{(\tau
)}_{k+1-x}(4t)\bigr)\nonumber\\
&&{}+\sum_{k,\ell\geq n}Q_k Q_\ell\bigl(K^{(\tau)}(0)_{k+x,\ell
}+K^{(\tau
)}(0)_{k-x,\ell}\bigr).
\nonumber
\end{eqnarray}
Remember $H_n(0)=\det(\Id-K(0))_{\ell^2(n,n+1,\ldots)}$.
%
\begin{lemma}\label{LemmaExtendedWithE}
Given the notation (\ref{eqE1234}) for the $E_i$'s, the extended kernel
$\widetilde\BK_m^{\mathrm{ext}}$ is given by
%
\begin{eqnarray}\label{eq35ext}
&&\frac{(-1)^{x_2} e^{4t_2}}{(-1)^{x_1} e^{4t_1}} \frac
{H_{n+1}(0)}{H_{n}(0)} \widetilde\BK_m^{\mathrm{ext}}(t_1,x_1;t_2,x_2)\nonumber\\
&&\qquad=-\Id_{[t_2<t_1]} p_{t_1-t_2}(x_1,x_2)\frac{H_{n+1}(0)}{H_{n}(0)}
+C(t_1-t_2,x_1-x_2)\nonumber\\[-8pt]\\[-8pt]
&&\qquad\quad{}+\frac{1}{(2\pi\I)^2}\oint_{\Gamma_0}dz \oint_{\Gamma_{0,z}} dw
\frac
{1}{z-w} \sum_{i=1}^4 E_i(z,w)\nonumber\\
&&\hspace*{11pt}\qquad\quad{}\times\biggl({\frac{(-w)^{x_2-1}}{(-z)^{x_1}}\frac
{e^{-t_1(z+z^{-1}+2)}}{e^{-t_2(w+w^{-1}+2)}}}
+ {\frac
{(-z)^{x_2}}{(-w)^{x_1+1}}\frac{e^{-t_1(w+w^{-1}+2)}}{e^{-t_2(z+z^{-1}+2)}}}
\biggr).\nonumber
\end{eqnarray}
\end{lemma}
%
\begin{theorem}\label{ThmExtKernel}
The extended kernel $\widetilde\BK_m^{\mathrm{ext}}$ is also expressed
as
%
\begin{eqnarray}\label{Kext}
&&\frac{(-1)^{x_2} e^{4t_2}}{(-1)^{x_1} e^{4t_1}}\frac
{H_{n+1}(0)}{H_{n}(0)} \widetilde\BK_m^{\mathrm{ext}}(t_1,x_1;t_2,x_2)
\nonumber\\
&&\qquad=-\Id_{[t_2<t_1]} p_{t_1-t_2}(x_1,x_2) \frac
{H_{n+1}(0)}{H_{n}(0)}+C(t_1-t_2,x_1-x_2)\nonumber\\
&&\qquad\quad{}+\sum_{c\geq0}
\bigl(A(t_1,x_1-c)A(-t_2,x_2-c)+A(t_1,-x_1-c)A(-t_2,-x_2-c)\nonumber\\[-8pt]\\[-8pt]
&&\hspace*{29pt}\qquad\quad{}-A(t_1,x_1-c)B(-t_2,x_2-c)-A(t_1,-x_1-c)B(-t_2,-x_2-c)\nonumber\\
&&\hspace*{29pt}\qquad\quad{}-B(t_1,x_1-c)A(-t_2,x_2-c)-B(t_1,-x_1-c)A(-t_2,-x_2-c)\bigr)\nonumber
\\
&&\qquad\quad{}-\sum_{c<0}
\bigl(B(t_1,x_1-c)B(-t_2,x_2-c)+B(t_1,-x_1-c)B(-t_2,-x_2-c)
\bigr).\nonumber
\end{eqnarray}
\end{theorem}
\begin{pf*}{Proofs of Lemma~\ref{LemmaExtendedWithE} and
Theorem~\ref{ThmExtKernel}}
First of all, let us focus on the term $(e^{(t_2-t_1){\cal
H}})(x_1,x_2)$ in (\ref{Kext0}). Remember that $t_2-t_1<0$, so we
can rewrite
%
\begin{eqnarray}
\frac{(-1)^{x_2} e^{4t_2}}{(-1)^{x_1} e^{4t_1}}
\bigl(e^{(t_2-t_1){\cal
H}}\bigr)(x_1,x_2) &=& \frac{(-1)^{x_2} e^{2t_2}}{(-1)^{x_1} e^{2t_1}}
I_{|x_1-x_2|}\bigl(2(t_2-t_1)\bigr)\nonumber\\
&=& \frac{e^{2t_2}}{e^{2t_1}} I_{|x_1-x_2|}\bigl(2(t_1-t_2)\bigr) \\
&=& p_{t_1-t_2}(x_1,x_2),
\nonumber
\end{eqnarray}
where we used the property $I_{n}(-2t)=(-1)^n I_n(2t)$ of the modified
Bessel function; see (\ref{eqbesselI}).

Next we derive the double integrals in (\ref{eq35ext}). The
corresponding expression of the kernel $\widetilde\BK_m$ in (\ref
{eq35c}) is a linear combination [not forgetting the conjugation
factor of the left-hand side of (\ref{eq35c})] of
%
\begin{equation}\label{eq527}
-\frac{w^{x_2-1}}{z^{x_1}}-\frac{z^{x_2}}{w^{x_1+1}}.
\end{equation}
Applying $e^{-t_1 {\cal H}}$ to the left and $e^{t_2 {\cal H}}$ to the
right, (\ref{eq527}) transforms into
%
\begin{equation}\label{eq527B}
-\frac{w^{x_2-1}}{z^{x_1}}\frac
{e^{-t_1(z+z^{-1}-2)}}{e^{-t_2(w+w^{-1}-2)}}-\frac
{z^{x_2}}{w^{x_1+1}}\frac{e^{-t_1(w+w^{-1}-2)}}{e^{-t_2(z+z^{-1}-2)}}.
\end{equation}
The multiplication by the prefactor $\frac{(-1)^{x_2}
e^{4t_2}}{(-1)^{x_1} e^{4t_1}}$ leads then to the expression in (\ref
{eq35ext}).

Next derive the terms with the sums in (\ref{Kext}) and the expression
for $C$. We act with the semigroup on the summation part of the
kernel (\ref{eqpropDecomp}), which is expressed in terms of
$A(x),B(x),C(x)$, namely
%
\begin{eqnarray}\qquad
\frac{H_{n+1}(0)}{H_{n}(0)}
\widetilde\BK_m(x_1,x_2)&=&\sum_{c\geq
0}[(-1)^{x_1}A(x_1-c)][(-1)^{x_2}A(x_2-c)]+\cdots\nonumber\\[-8pt]\\[-8pt]
&&{}+(-1)^{x_1-x_2}C(x_1-x_2)
\nonumber
\end{eqnarray}
with $A(x), B(x), C(x)$ given in Proposition~\ref{LemmaABCD}. So,
except for the term $C(x_1-x_2)$, the expression above is a sum of
decoupled terms. Therefore acting on the $(-1)^x A(\pm x-c)$'s and
$(-1)^x B(\pm x-c)$'s with $e^{-t_1 {\cal H}}$ to the left amounts (by
linearity) to acting on the $(-1)^x J_{N\pm x}(2t)$ (for some $N$
depending on the terms) and finally to acting on $1/(-z)^{\pm x}$
inside the integration. More precisely, by (\ref{S1}) with
$f(x):=1/(-z)^{\pm x}$, we have
%
\begin{eqnarray}
(e^{-t_1 {\cal H}} f)(x)&=&e^{t_1(z+z^{-1}+2)} f(x)
\quad\mbox{and}\nonumber\\[-8pt]\\[-8pt]
\quad (f e^{t_2 {\cal H}})(x)&=&
e^{-t_2(z+z^{-1}+2)}f(x),\nonumber
\end{eqnarray}
from which, by linearity,
%
\begin{eqnarray}
\sum_{y\in\Z}(e^{-t_1 {\cal H}})(x,y) (-1)^y J_{N\pm y}(2t)&=&\oint
_{\Gamma_0} \frac{dz}{2\pi\I z} \frac{e^{t(z-z^{-1})}}{z^N} (e^{-t_1
{\cal H}} f)(x)\nonumber\\
&=&\oint_{\Gamma_0} \frac{dz}{2\pi\I z} \frac{e^{t(z-z^{-1})}}{z^N
(-z)^{\pm x}} e^{t_1(z+z^{-1}+2)}\\
&=&(-1)^x e^{4 t_1} J_{N\pm x}^{(t_1)}(2t)
\nonumber
\end{eqnarray}
and
%
\begin{equation}
\sum_{y\in\Z}(-1)^y J_{N\pm y}(2t)(e^{t_2 {\cal H}})(y,x) =(-1)^x e^{-4
t_2} J_{N\pm x}^{(-t_2)}(2t).
\end{equation}
This extends to the functions $(-1)^x A(\pm x-c)$, $(-1)^x B(\pm x-c)$
because they are linear in the $(-1)^x J_{N\pm x}(2t)$ [see (\ref
{eq314})]. Explicitly, applying $e^{-t_1 {\cal H}}$ (to the left) to
$(-1)^x A(\pm x-c)$ amounts to replacing $A(\pm x-c)$ with $e^{4
t_1}A(t_1,\break\pm x-c)$. Similarly, applying $e^{t_2 {\cal H}}$ (to the
right) to $(-1)^x A(\pm x-c)$ amounts to replacing $A(\pm x-c)$ with
$e^{-4t_2}A(-t_2,\pm x-c)$. The same holds for $B$ instead of $A$. Thus
we have obtained the terms in kernel (\ref{Kext}) including $A$'s and $B$'s.

Exactly the same procedure applies for the term
$(-1)^{x_1-x_2}C(x_1-x_2)$, because it is again a linear combination of
$(-1)^{x_1-x_2} J_{N\pm x_1\mp x_2}(4t)$. Therefore acting with
$e^{-t_1 {\cal H}}$ and $e^{t_2 {\cal H}}$ as before on
$(-1)^{x_1-x_2}C(x_1-x_2)$ leads to the replacement of $C(x_1-x_2)$ by
$e^{4(t_1-t_2)}C(t_1-t_2,x_1-x_2)$. This completes the proof of
formulas (\ref{eq35ext}) and (\ref{Kext}) for the extended kernel,
thus establishing Lemma~\ref{LemmaExtendedWithE} and Theorem \ref
{ThmExtKernel}.
\end{pf*}

\section{Asymptotics}\label{SectAsymptotics}
In this section we prove the first half of Theorem \ref
{ThmExtKernelAsympt}, namely formula (\ref{ExtKernelA}). From the
discussion in Section~\ref{Model} after Theorem~\ref{MainTheorem},
concerning the interaction between the top and bottom sets of random
walks, we rescale space, time and the gap $n=2m+1$ between the two
groups of walkers, as follows:
%
\begin{equation}\label{eqScalingXY}
m=2t+\sigma t^{1/3},\qquad x_i=\xi_i t^{1/3},\qquad t_i=s_i
t^{2/3},\qquad i=1,2,
\end{equation}
where $\sigma\in\R$ is a fixed parameter modulating the ``strength of
interaction'' between the upper and lower sets of walks. To prove
formula (\ref{ExtKernelA}) of Theorem~\ref{ThmExtKernelAsympt}, we
first analyze the asymptotics of the building blocks and determine some
bounds which will be used later to show that we can exchange (by
dominated convergence) the large time limit with the integrals (sums).

Recall from (\ref{eq613}), (\ref{Q}) and (\ref{55}) the functions
$J^{(\tau)}_{x}(2t)$ and ${\cal Q}$, and the kernel $K^{(\tau
)}(0)_{k,\ell} $,
%
\begin{eqnarray} \label{63}
J^{(\tau)}_{x}(2t)&=&e^{-2\tau}\biggl(\frac{t+\tau}{t-\tau}\biggr)^{x/2}
J_{x}\bigl(2\sqrt{t^2-\tau^2}\bigr),\nonumber\\
K^{(\tau)}(0)_{k,\ell} &=& \sum_{a\geq0} J^{(\tau
)}_{a+k+1}(4t)J_{a+\ell
+1}(4t),\\
{\cal Q}(\kappa) &=& [(\Id-\chi_{\tilde\sigma} K_{\mathrm{Ai}} \chi
_{\tilde
\sigma})^{-1}\chi_{\tilde\sigma} {\Ai}](\kappa)\qquad
\mbox{with }
\tilde \sigma:=2^{2/3}\sigma,\nonumber
\end{eqnarray}
and where $\chi_{a}(x)=\Id_{[x>a]}$.
Remember from (\ref{Airy}) the definition of
%
\begin{equation}
\Ai^{(s)} (\xi):= e^{\xi s+(2/3) s^3}\Ai(\xi+ s^2)
\end{equation}
and define the Airy-like kernel
%
\begin{equation}
K^{(s)}_{\mathrm{Ai}}(\kappa,\lambda):=\int_0^{\infty} d\gamma \Ai
^{(s2^{-2/3})}(\kappa+\gamma)
\Ai(\lambda+\gamma).
\end{equation}
Also define the following step functions of $\kappa,\lambda\in\R$, for
which---by anticipation---we indicate the limits for $t\to\infty$:
%
\begin{eqnarray}\label{eq49}
{\cal J}^{(s)}_{t}(\kappa)&:=&t^{1/3} J^{(s t^{2/3})}_{[2t+\kappa
t^{1/3}+1]}(2t) \rightarrow\Ai^{(s)}(\kappa),\nonumber\\
\hspace*{18pt}{\cal K}^{(s)}_t(\kappa,\lambda
)&:=&(2t)^{1/3}K^{(st^{2/3})}(0)_{[4t+\kappa(2t)^{1/3}],[4t+\lambda
(2t)^{1/3}]} \rightarrow K^{(s)}_{\mathrm{Ai}}(\kappa,\lambda
),\nonumber\\
{\cal Q}_t(\kappa)&:=&(2t)^{1/3}Q_{[4t+\kappa(2t)^{1/3}]}\\
&=&\bigl[\bigl(\Id-\chi_{({n-4t})/{(2t)^{1/3}}} {\cal K}^{(0)}_t
\chi
_{({n-4t})/{(2t)^{1/3}}}\bigr)^{-1} \chi_{({n-4t})/{(2t)^{1/3}}}
{\cal J}^{(0)}_{2t}\bigr](\kappa) \nonumber\\
&\rightarrow&{\cal
Q}(\kappa).\nonumber
\end{eqnarray}

\begin{lemma}\label{LemBoundsJandK}
We have the following bounds and limits for ${\cal J}^{(s)}_{t}$ and
${\cal K}^{(s)}_t$ defined in (\ref{eq49}). There exists a $t_0>0$
such that uniformly for $t\geq t_0$ it holds that
%
\begin{equation}\label{eq416}
\bigl|{\cal J}^{(s)}_{t}(\kappa)\bigr|
\leq c_1 \min\{1,e^{-\theta\kappa}\}
,\qquad
\bigl|{\cal K}^{(s)}_t(\kappa,\lambda)\bigr|\leq c_2 e^{-\theta(\kappa
+\lambda)}
\end{equation}
for any fixed $\theta>0$ and some constants $c_1,c_2>0$ (independent
of $t$).
Moreover
%
\begin{equation}\label{L0}
\lim_{t\to\infty} {\cal J}^{(s)}_{t}(\kappa)=
\Ai^{(s)}(\kappa),\qquad \lim_{t\to\infty} {\cal
K}^{(s)}_{t}(\kappa
,\lambda)=K_{\mathrm{Ai}}^{(s)}(\kappa,\lambda)
\end{equation}
uniformly for $\kappa, \lambda$, and $s$ in a bounded set.
\end{lemma}
\begin{pf}
We have
%
\begin{eqnarray}\label{67}
{\cal J}^{(s)}_{t}(\xi)&=& t^{1/3} J^{(s t^{2/3})}_{[2t+\xi t^{1/3}]}(2t)
\nonumber\\
&=& e^{-2s t^{2/3}}\biggl(\frac{1+s t^{-1/3}}{1-s t^{-1/3}}
\biggr)^{
(1/2)[2t+ \xi t^{1/3}]} \\
&&{}\times t^{1/3} J_{[2t+\xi
t^{1/3}]}\bigl(2t\sqrt{1-s^2t^{-2/3}}\,\bigr).\nonumber
\end{eqnarray}
The prefactor can be estimated for $t\to\infty$, as follows:
%
\begin{equation}\label{eq626}
e^{-2st^{2/3}}\biggl(\frac{1+s t^{-1/3}}{1-s t^{-1/3}}
\biggr)^{t+
(1/2) \xi t^{1/3}} = e^{\xi s+(2/3) s^3}\bigl(1+\Or(t^{-1/3})\bigr),
\end{equation}
where the $\Or(t^{-1/3})$ is uniform for $s$ in a bounded set and
independent of~$\xi$.
Therefore, for $t$ large enough, $|(\ref{eq626})|\leq\exp(2|\xi s|
+|s^3|)$. Concerning the remaining part of (\ref{67}), using (\ref
{eqA1}), one readily obtains
%
\begin{equation}
\lim_{t\to\infty}t^{1/3} J_{[2t+\xi t^{1/3}]}\bigl(2t\sqrt
{1-s^2t^{-2/3}}\,\bigr) = \Ai(\xi+s^2).
\end{equation}
Regarding the bound, for $s$ in a bounded set, if $t$ is large enough
it follows from bound (\ref{eqA3}) that
%
\begin{equation}
\bigl|t^{1/3} J_{[2t+\xi t^{1/3}]}\bigl(2t\sqrt{1-s^2t^{-2/3}}\,
\bigr)\bigr|
\end{equation}
is first of all uniformly bounded and for large $\xi$ it decays as
$e^{-\beta\xi}$ for any choice of $\beta>0$. The statements in the
first parts of (\ref{eq416}) and (\ref{L0}) then follow if we choose
$\beta$ satisfying $\beta\geq\theta+2|s|$ for any $s$ in the given
bounded set.

To compute the limit of ${\cal K}^{(s)}_t$, one uses definition (\ref
{eq49}) and formula (\ref{63}) for $K^{(st^{2/3})}(0)$, but with $J$
replaced by ${\cal J}$ in the last equality below,
%
\begin{eqnarray}\label{eq417}\quad
{\cal K}^{(s)}_t(\kappa,\lambda) &=& (2t)^{1/3}K^{(st^{2/3})}(0)
_{[4t+\kappa(2t)^{1/3}],[4t+\lambda(2t)^{1/3}]} \nonumber\\
&=& (2t)^{1/3}\sum_{\gamma\in(2t)^{-1/3}\N}
J^{(s2^{-2/3} (2t)^{2/3})}_{[4t+(\gamma+\kappa)(2t)^{1/3}]}(4t)
J_{[4t+(\gamma+\lambda) (2t)^{1/3}]}(4t)\\
&=& \frac{1}{(2t)^{1/3}}\sum_{\gamma\in(2t)^{-1/3} \N} {\cal
J}^{(s2^{-2/3})}_{2t}(\kappa+\gamma){\cal J}^{(0)}_{2t}(\lambda
+\gamma).
\nonumber
\end{eqnarray}
From this, using bound (\ref{eq416}) on $\cal J$, we obtain
%
\begin{equation}
\bigl|{\cal K}^{(s)}_t(\kappa,\lambda)\bigr|\leq c_1^2 e^{-\theta(\kappa
+\lambda
)} \frac{1}{(2t)^{1/3}}\sum_{\gamma\in(2t)^{-1/3} \N} e^{-2\theta
\gamma
}\leq c_2 e^{-\theta(\kappa+\lambda)}
\end{equation}
for $t\geq t_0=1$ and some $c_2>0$, uniformly for $s$ in a bounded set.

We can think of the sum in (\ref{eq417}) as an integral of piece-wise
constant functions. The first bound in (\ref{eq416}) allows us to use
dominated convergence to exchange the limit and the integral. Then,
$\lim_{t\to\infty} {\cal J}^{(s)}_{t}(\kappa)=\Ai^{(s)}(\kappa)$ yields
%
\begin{equation}\quad
\lim_{t\to\infty} {\cal K}^{(s)}_t(\kappa,\lambda) =\int
_0^{\infty}
d\gamma \Ai^{(2^{-2/3}s)}(\kappa+\gamma) \Ai(\lambda+\gamma
)=K^{(s)}_{\mathrm{Ai}}(\kappa,\lambda).
\end{equation}
\upqed\end{pf}
%
\begin{lemma}\label{lemMt}
Set $\tilde\sigma_t:=\frac{n-4t}{(2t)^{1/3}}$ and define the operator
${\cal M}_t=\chi_{\tilde\sigma_t} {\cal K}^{(0)}_t \chi_{\tilde
\sigma
_t}$, appearing in the definition (\ref{eq49}) of ${\cal Q}_t$. Then,
uniformly for $t\geq t_0$, we have for the operator-norm\footnote{Where
$\|A\|=\sup_{|f|\leq1} |A f|$.} \mbox{$\| \cdot\|$},
%
\begin{equation}
\|{\cal M}_t\|<1,
\end{equation}
which implies that
%
\begin{equation}
\|(\Id-{\cal M}_t)^{-1}\|\leq(1-\|{\cal M}_t\|)^{-1}\leq C<\infty
\end{equation}
for some finite constant $C$ independent of $t$.
\end{lemma}
\begin{pf}
By Lemma~\ref{LemBoundsJandK} and the fact that $\tilde\sigma_t\to
\tilde
\sigma$ as $t\to\infty$, it follows that
%
\begin{equation}\label{eq422}
\lim_{t\to\infty}{\cal M}_t=\chi_{\tilde\sigma} K_{\mathrm{Ai}} \chi
_{\tilde
\sigma}=:{\cal M}
\end{equation}
pointwise. Moreover,
%
\begin{eqnarray}\label{eq421}
{\lim_{t\to\infty}}\|{\cal M}_t-{\cal M}\|^2&\leq&{\lim_{t\to\infty
}}\|
{\cal M}_t-{\cal M}\|_{\mathrm{HS}}^2\nonumber\\
&=&\lim_{t\to\infty}\int d\kappa\,
d\lambda
|{\cal M}_t(\kappa,\lambda)-{\cal M}(\kappa,\lambda)
|^2\\
&=&\int d\kappa \,d\lambda{\lim_{t\to\infty}}|{\cal M}_t(\kappa
,\lambda
)-{\cal
M}(\kappa,\lambda)|^2=0,\nonumber
\end{eqnarray}
where we use by Lemma~\ref{LemBoundsJandK} dominated convergence to
exchange the limit and the integral together with (\ref{eq422}). It is
known that $\lambda_{\max}=\|{\cal M}\|<1$ for any fixed $\tilde
\sigma$
(see, e.g.,~\cite{TW94}). This, together with (\ref{eq421}), implies that
%
\begin{equation}
\|{\cal M}_t\|\leq\|{\cal M}\|+\|{\cal M}_t-{\cal M}\|<1
\end{equation}
for $t$ large enough.
\end{pf}
%
\begin{lemma}\label{LemBoundsQ}
Consider ${\cal Q}_t$ as defined in (\ref{eq49}). There exists a
$t_0>0$ such that, uniformly for $t\geq t_0$, it holds
%
\begin{equation}\label{eq419}
|{\cal Q}_{t}(\kappa)|\leq c_3 e^{-\theta\kappa}
\end{equation}
for any $\theta>0$ and some constant $c_3>0$ (independent of $t$). Moreover,
%
\begin{equation}\label{Qlim}
\lim_{t\to\infty} {\cal Q}_{t}(\kappa)={\cal Q}(\kappa)
\end{equation}
uniformly for $\kappa$ in a bounded set.
\end{lemma}
\begin{pf}
For the sake of this proof, set ${\cal J}_t:={\cal J}^{(0)}_t$ and
${\cal K}_t:={\cal K}^{(0)}_t$. First of all we prove that ${\cal
Q}_t(\kappa)$ is uniformly bounded for $t\geq t_0$. Recall that ${\cal
Q}_t(\kappa)=[(\Id-{\cal M}_t)^{-1} \chi_{\tilde\sigma_t}
J_{2t}](\kappa
)$. Since $(\Id-{\cal M}_t)^{-1}$ exists, we can use the identity
%
\begin{equation}
(\Id-{\cal M}_t)^{-1} = \Id+ \chi_{\tilde\sigma_t}{\cal K}_t \chi
_{\tilde\sigma_t}(\Id-{\cal M}_t)^{-1},
\end{equation}
which upon integrating from $\tilde\sigma$ to $\infty$ against the
function ${\cal J}_{2t}$ gives
%
\begin{equation}
{\cal Q}_t(\kappa)=\chi_{\tilde\sigma_t} {\cal J}_{2t}(\kappa) +
\int
_{\tilde\sigma_t}^\infty d\lambda\, {\cal K}_t(\kappa,\lambda)
[(\Id
-{\cal M}_t)^{-1}\chi_{\tilde\sigma_t} {\cal J}_{2t}](\lambda).
\end{equation}
Thus,
%
\begin{equation}\qquad
|{\cal Q}_t(\kappa)|\leq|\chi_{\tilde\sigma_t} {\cal
J}_{2t}(\kappa
)|+\int_{\tilde\sigma_t}^\infty d\lambda |{\cal K}_t(\kappa
,\lambda)| |[(\Id-{\cal M}_t)^{-1}\chi_{\tilde\sigma
_t} {\cal
J}_{2t}](\lambda)|.
\end{equation}
But
%
\begin{equation}
|[(\Id-{\cal M}_t)^{-1}\chi_{\tilde\sigma_t} {\cal J}_{2t}](\lambda
)|\leq\|(\Id-{\cal M}_t)^{-1}\| |{\cal J}_{2t}|_\infty
\end{equation}
is uniformly bounded for $t\geq t_0$ (by Lemmas~\ref{LemBoundsJandK} and
\ref{lemMt}). Then, using the bound for ${\cal K}_t$ and ${\cal
J}^{(s)}_{t}(\kappa)$ in (\ref{eq416}) we obtain the bound
(\ref{eq419}).\vadjust{\goodbreak}

To prove (\ref{Qlim}), we show that
%
\begin{equation}\label{eq4C}
|{\cal Q}_t-{\cal Q}|_\infty= \sup_{\kappa} |{\cal Q}_t(\kappa
)-{\cal
Q}(\kappa)|\to0
\end{equation}
as $t\to\infty$. We have
%
\begin{eqnarray}\label{eq431}
|{\cal Q}_t-{\cal Q}|_\infty&=&|(\Id-{\cal M}_t)^{-1} \chi
_{\tilde
\sigma_t}{\cal J}_{2t}-(\Id-{\cal M})^{-1} {\chi_{\tilde\sigma}\Ai}
|_\infty\nonumber\\
&\leq&| [(\Id-{\cal M}_t)^{-1} -(\Id-{\cal M})^{-1}
]\chi_{\tilde\sigma}{\cal J}_{2t}|_\infty\\
&&{}+|(\Id-{\cal M})^{-1}[\chi_{\tilde\sigma}{\cal
J}_{2t}-\chi_{\tilde\sigma}{\Ai}]|_\infty+\Or
(t^{-1/3}),\nonumber
\end{eqnarray}
where the correction term $\Or(t^{-1/3})$ comes from the fact that the
difference between $\tilde\sigma_t$ and $\tilde\sigma$ is not larger
than $(2t)^{-1/3}$.
Then,
%
\begin{eqnarray}
\mbox{(\ref{eq431})}&\leq&\|(\Id-{\cal M}_t)^{-1} -(\Id-{\cal
M})^{-1}\| |\chi_{\tilde\sigma}{\cal J}_{2t}
|_\infty
\nonumber\\[-8pt]\\[-8pt]
&&{}+\|(\Id-{\cal M})^{-1}\| |\chi_{\tilde\sigma
}{\cal
J}_{2t}-\chi_{\tilde\sigma}{\Ai}|_\infty
+\Or(t^{-1/3}).
\nonumber
\end{eqnarray}
The first term goes to zero as $t\to\infty$. Indeed, $|\chi
_{\tilde
\sigma}{\cal J}_{2t}|_\infty\leq C<\infty$ by Lem\-ma~\ref
{LemBoundsJandK}, and, using the identity
%
\begin{equation}\quad
(\Id-{\cal M}_t)^{-1} -(\Id-{\cal M})^{-1} =(\Id-{\cal
M}_t)^{-1}
[{\cal M}_t-{\cal M}](\Id-{\cal M})^{-1}
\end{equation}
together with the fact that $\|{\cal M}_t\|<1$, $\|{\cal M}\|<1$, and
$\|{\cal M}-{\cal M}_t\|\to0$ in the $t\to\infty$ limit [see
Lemma~\ref{lemMt} and (\ref{eq421})]; so one has $\|(\Id
-{\cal
M}_t)^{-1} -(\Id-{\cal M})^{-1}\|\to0$. The second term goes to
zero as well, since $\|(\Id-{\cal M})^{-1}\|$ is bounded
and, by Lemma~\ref{LemBoundsJandK},
$|\chi_{\tilde\sigma}{\cal J}_{2t}-\chi_{\tilde\sigma}{\Ai}
|_\infty\to0$.
\end{pf}
\begin{pf*}{Proof of Theorem~\ref{ThmExtKernelAsympt}, formula
(\ref{ExtKernelA})}
We now define new functions ${\cal A}_t(s,\xi)$, ${\cal B}_t(s,\xi)$,
${\cal C}_t(s,\xi)$, which are rescaled versions of $A(\tau,x)$,
$B(\tau
,x)$, $C(\tau,x)$ [see formula (\ref{53})] under the scaling (\ref
{eqScalingXY}):
%
\begin{eqnarray}
{\cal A}_t(s,\xi)&:=&t^{1/3} A(st^{2/3},\xi t^{1/3}),\nonumber\\
{\cal B}_t(s,\xi)&:=&t^{1/3} B(st^{2/3},\xi t^{1/3}),\\
{\cal C}_t(s,\xi)&:=& t^{1/3} C(st^{2/3},\xi t^{1/3}).\nonumber
\end{eqnarray}
As $t\to\infty$, these functions will converge to ${\cal A}(s,\xi),
{\cal B}(s,\xi), {\cal C}(s,\xi)$ of (\ref{eqSpaceTimeAB}) and (\ref
{eqSpaceTimeC}).

One then recognizes in these expressions functions (\ref{eq49}), thus yielding
%
\begin{eqnarray}\qquad
{\cal A}_t(s,\xi)&=&{\cal J}^{(s)}_{t}(\sigma-\xi)\nonumber\\
&&{}+\frac{1}{(2t)^{1/3}}\sum_{\kappa\in I_{n,t}}\frac
{1}{(2t)^{1/3}}\sum
_{\alpha\in(2t)^{-1/3} \N}
{\cal Q}_t(\kappa) {\cal J}^{(0)}_{2t}(\kappa+\alpha)\nonumber\\
&&\hspace*{151pt}{}\times {\cal
J}^{(s)}_{t}(2^{1/3}\alpha+\sigma-\xi),\\
{\cal B}_t(s,\xi)&=&\frac{1}{(2t)^{1/3}}\sum_{\kappa\in
I_{n,t}}{\cal
Q}_t(\kappa) {\cal
J}^{(s)}_t(\xi-\sigma+2^{1/3}\kappa-t^{-1/3}),\nonumber\\[-1pt]
{\cal C}_t(s,\xi)&=&\frac{2^{-1/3}}{(2t)^{1/3}}\sum_{\kappa\in
I_{n,t}}{\cal Q}_t(\kappa)
\bigl({\cal J}^{(2^{-2/3}s)}_{2t}(\kappa-2^{-1/3}\xi)\nonumber\\[-1pt]
&&\hspace*{89pt}{} + {\cal
J}^{(2^{-2/3}s)}_{2t}(\kappa+2^{-1/3}\xi)\bigr)\nonumber\\[-1pt]
&&{}+\frac{2^{-1/3}}{(2t)^{2/3}}\sum_{\kappa,\lambda\in I_{n,t}}{\cal
Q}_t(\kappa) {\cal Q}_t(\lambda)
\bigl({\cal K}^{(s)}_{t}(\kappa-2^{-1/3}\xi,\lambda)\nonumber\\[-1pt]
&&\hspace*{135pt}{}+{\cal
K}^{(s)}_{t}(\kappa+2^{-1/3}\xi,\lambda)\bigr).\nonumber
\end{eqnarray}
For instance, the function $J_{k+1+a}(4t)$ in $A(\tau,x)$ becomes, upon
setting $a=\alpha(2t)^{1/3}$ and $\kappa:=(2t)^{-1/3}(k-4t)$,
%
\begin{equation}
J_{k+1+a}(4t)=J_{[4t+(\kappa+\alpha
)(2t)^{1/3}+1]}(4t)=(2t)^{-1/3}{\cal
J}_{2t}^{(0)}(\kappa+\alpha).
\end{equation}
Notice that the sum over $k\geq n$ in the expressions (\ref{53})
becomes a sum over $\kappa\in I_{n,t}$ with
%
\begin{equation}
I_{n,t}:=(2t)^{-1/3}(\{n,n+1,\ldots\}-4t),
\end{equation}
so that the condition $k\geq n=2m+1=4t+2 \sigma t^{1/3}+1$ translates into
$\kappa=(2t)^{-1/3}(k-4t)> 2^{2/3}\sigma=\tilde\sigma$.
Setting the summation variable $c=\gamma t^{1/3}$, rewrite the
kernel (\ref{Kext}) in Theorem~\ref{ThmExtKernel}, with the scaling
(\ref{eqScalingXY})
%
\begin{eqnarray}\label{eq47}
&&\frac{(-1)^{x_2} e^{4t_2}}{(-1)^{x_1} e^{4t_1}}\frac
{H_{n+1}(0)}{H_{n}(0)} \widetilde\BK_m^{\mathrm{ext}}(t_1,x_1;t_2,x_2)\nonumber\\[-1pt]
&&\qquad = -\Id_{[s_1>s_2]} \frac{H_{n+1}(0)}{H_{n}(0)} t^{1/3}
p_{(s_1-s_2)t^{2/3}}(\xi_1t^{1/3},\xi_2t^{1/3})+ {\cal
C}_t(s_1-s_2,\xi
_1-\xi_2)\nonumber\\[-1pt]
&&\qquad\quad{}+\frac{1}{t^{1/3}}
\sum_{\gamma\in t^{-1/3}\N}
\bigl({\cal A}_t(s_1,\xi_1-\gamma){\cal A}_t(-s_2,\xi_2-\gamma)\nonumber\\[-1pt]
&&\hspace*{105pt}{}+{\cal A}_t(s_1,-\xi_1-\gamma){\cal
A}_t(-s_2,-\xi_2-\gamma)\nonumber\\[-1pt]
&&\hspace*{105pt}{}- {\cal A}_t(s_1,\xi_1-\gamma){\cal
B}_t(-s_2,\xi_2-\gamma)\nonumber\\[-8pt]\\[-8pt]
&&\hspace*{105pt}{}-{\cal
A}_t(s_1,-\xi_1-\gamma){\cal
B}_t(-s_2,-\xi_2-\gamma)\nonumber\\[-1pt]
&&\hspace*{105pt}{}-{\cal B}_t(s_1,\xi_1-\gamma){\cal A}_t(-s_2,\xi_2-\gamma)\nonumber\\[-1pt]
&&\hspace*{107pt}{}-{\cal
B}_t(s_1,-\xi_1-\gamma){\cal A}_t(-s_2,-\xi_2-\gamma)\bigr)\nonumber\\[-1pt]
&&\qquad\quad{}-\frac{1}{t^{1/3}}\sum_{\gamma\in t^{-1/3}\Z_-} \bigl({\cal
B}_t(s_1,\xi
_1-\gamma){\cal B}_t(-s_2,\xi_2-\gamma)\nonumber\\[-1pt]
&&\hspace*{110pt}{}+{\cal B}_t(s_1,-\xi
_1-\gamma
){\cal
B}_t(-s_2,-\xi_2-\gamma)\bigr).\nonumber
\end{eqnarray}
In view of (\ref{eq313}) we have $\lim_{t\to\infty
}H_{n+1}(0)/H_n(0)=1$ and in the $t\to\infty$ limit,
$(n-4t)/(2t)^{1/3}\to\tilde\sigma$. Notice that the sums with the
preceding volume element, $1/t^{1/3}$ or $1/(2t)^{1/3}$ depending on
the case, can be just thought of as integrals with the integrand being
piece-wise constant. What follows holds uniformly in $t$ for $t\geq
t_0$ where $t_0$ is a fixed constant.
The exponential bounds of Lemmas~\ref{LemBoundsJandK} and \ref
{LemBoundsQ} imply that for any $\theta>0$ there exists some $c>0$ (the
constant $c$ depends on $\sigma$, which is, however, fixed)
%
\begin{equation}\label{eq412}
|{\cal A}_t(s,-\xi)|\leq c e^{-\theta\xi} \quad\mbox{and}\quad \lim
_{t\to\infty} {\cal A}_t(s,\xi)={\cal A}(s,\xi).
\end{equation}
Moreover ${\cal A}_t(s,\xi)$ tends to ${\cal A}(s,\xi)$ uniformly on
bounded sets, by uniform convergence on bounded sets and dominated
convergence of the integrand.
Using the exponential bound of Lemma~\ref{LemBoundsQ} and the fact that
${\cal J}_t$ is just bounded, we obtain similarly
%
\begin{equation}\label{eq413}
|{\cal B}_t(s,\xi)|\leq c\min\{1,e^{-\theta\xi}\}
\quad\mbox{and}\quad
\lim_{t\to\infty} {\cal B}_t(s,\xi)={\cal B}(s,\xi).
\end{equation}
Finally, the exponential bounds of Lemmas~\ref{LemBoundsJandK} and
\ref{LemBoundsQ} imply that
%
\begin{equation}\label{eq414}
|{\cal C}_t(s,\xi)|\leq c \quad\mbox{and}\quad \lim_{t\to\infty}
{\cal
C}_t(s,\xi)={\cal C}(s,\xi),
\end{equation}
where the last limit holds uniformly for $\xi$ and $s$ in bounded sets.

Using the bounds in (\ref{eq412}), (\ref{eq413}) and (\ref{eq414}),
one concludes that the integrands (summands) in (\ref{eq47}) are
uniformly bounded by functions which are integrable (summable). This is
uniform for $\xi,\eta$ and $s$ in a bounded set. Then, by dominated
convergence, we can take the limit inside, thus yielding (\ref
{ExtKernelA}). Finally, the Gaussian term in (\ref{ExtKernelA}) comes
from the known asymptotic (for $s>0$),
%
\begin{equation}
\lim_{t\to\infty} t^{1/3}e^{-2st^{2/3}} I_{\xi t^{1/3}}(2s t^{2/3})=
\frac{1}{\sqrt{4\pi s}}\exp\bigl(-\xi^2/(4s)\bigr),
\end{equation}
which can be derived from a saddle point argument.
\end{pf*}

\section{Integral representation of the Tacnode kernel}\label
{IntegralReprKernel}
To derive the double integral representation (\ref{ExtKernelB}) of
Theorem~\ref{ThmExtKernelAsympt} there are two ways. One can use the
Airy functions integral representations (\ref{A4k}) together with
%
\begin{equation}
\int_0^\infty d\lambda\, e^{-\lambda(u-v)}=\frac{1}{u-v}
\qquad\mbox{whenever }\Re(u-v)>0.
\end{equation}
This is quite straightforward, but it requires several computations
which are not reported here.

The second is to do a steepest descent analysis starting from
formula (\ref{eq35ext}) in Lemma~\ref{MainTheorem}. Here we merely
indicate a sketch of the saddle point argument (not a proof). The
limits of the other terms have been discussed in the previous section.
The main task\vadjust{\goodbreak} here is to take the limit of this double integral, when
$t\to\infty$, with the scaling
%
\begin{eqnarray}\label{sc}
n&=&2m+1,\qquad m = 2t+ \sigma t^{1/3}, \nonumber\\
z&=&-1+\zeta t^{-1/3} \quad\mbox{and}\quad w=-1+\omega t^{-1/3},\\
x_i&=&\xi_it^{1/3} \quad\mbox{and}\quad t_i=s_it^{2/3},\qquad
i=1,2.\nonumber
\end{eqnarray}

Also recall the definitions (\ref{Laplace}) of the Laplace transforms
$\hat{\cal Q}(\zeta)$ and $\hat{\cal P}(\zeta)$, as well as the
function ${\cal C}$ in (\ref{eqSpaceTimeC}). The reader is reminded of
the steepest descent discussion in Section~\ref{shape1}.
For taking the limit of the extended kernel, we need the following lemma.
%
\begin{lemma} \label{Lemma71} Given the scaling (\ref{sc}) above, the
following limits hold:
%
\begin{equation}
\lim_{t\to\infty}e^{t(z-z^{-1})} (-z)^m=e^{{\zeta^3}/3-\sigma
\zeta}
\end{equation}
and
%
\begin{eqnarray}
\lim_{t\to\infty}T_n(z^{-1})&=&e^{-2\sigma\zeta} \hat{\cal
Q}(\zeta
),\qquad
\lim_{t\to\infty}T_n(w) =e^{2\sigma\omega} \hat{\cal Q}(-\omega
),\nonumber\\[-8pt]\\[-8pt]
\lim_{t\to\infty}S_n(z^{-1})&=&\hat{\cal P}(\zeta),\qquad
\lim
_{t\to\infty}S_n(w) =\hat{\cal P}(-\omega),
\nonumber
\end{eqnarray}
where $\hat{\cal P}$ and $\hat{\cal Q}$ are the Laplace transforms
defined in (\ref{Laplace}). One also checks
%
\begin{equation}\label{A3}
\lim_{t\to\infty}\frac{(-w)^{x_2-1}}{(-z)^{x_1}}=\frac{e^{\xi
_1\zeta
}}{e^{\xi_2\omega}} \quad\mbox{and}\quad
\lim_{t\to\infty}e^{-t_i(z+z^{-1}+2)} = e^{s_i\zeta^2}.
\end{equation}
\end{lemma}
\begin{pf}
Letting $t\rightarrow\infty$, setting $n=2m+1$, $m=2t+\sigma t^{1/3}$,
the critical point will be at $z,w=-1$, and thus the leading
contribution will come from the neighborhood of the critical points,
which suggests the scalings in $z$ and $w$ above. The Taylor expansion
of the $F$-function (\ref{Taylor}) gives
%
\begin{eqnarray}\label{A3k}
e^{t(z-z^{-1})} (-z)^m &=& e^{t(z-z^{-1})+m\log(-z)} =e^{tF(z)+\sigma
t^{1/3}\log(-z)}\nonumber\\
&=& e^{tF(-1+\zeta t^{-1/3})+\sigma t^{1/3}\log(1-\zeta t^{-1/3})}\\
&=&
e^{{\zeta^3}/3-\sigma\zeta}\bigl(1+\Or(t^{-1/3})\bigr).
\nonumber
\end{eqnarray}
Setting in addition the scaling for $t_i $ and $x_i$, one finds by
Taylor expanding about $z=-1$ and $w=-1$ the limits (\ref{A3}).
Introducing the running variable $k=4t+\kappa(2t)^{1/3}$, one gets
%
\begin{eqnarray}\label{A4}
\lim_{t\to\infty}T_n(z^{-1})&=&\lim_{t\to\infty}\sum_{k\geq n}
\frac
{Q_k}{(-z)^{k-n+1}}\nonumber\\
&=& \lim_{t\to\infty}\sum_{k\geq n} Q_k
e^{-(k-n+1)\log(-z)} \nonumber\\
&=&\lim_{t\to\infty}(2t)^{-1/3} \sum_{\kappa\geq\tilde\sigma
+(2t)^{-1/3}} (2t)^{1/3} {Q_{4t+\kappa(2t)^{1/3} }}\\
&&\hspace*{109pt}{}\times e^{-(\kappa-\tilde
\sigma)(2t)^{1/3}(-\zeta t^{-1/3})}\nonumber\\
&=&\int_{\kappa\geq\tilde\sigma} d\kappa\, {\cal Q}(\kappa)
e^{(\kappa-\tilde\sigma)\zeta2^{1/3}}
= e^{-2\sigma\zeta} \hat{\cal Q}(\zeta)
\nonumber
\end{eqnarray}
and similarly
%
\begin{equation}\label{A4b}
\lim_{t\to\infty}T_n(w) =e^{2\sigma\omega} \int_{\kappa\geq
\tilde
\sigma}
d\kappa \,{\cal Q}(\kappa) e^{-\kappa\omega2^{1/3}} = e^{2\sigma
\omega} \hat{\cal Q}(-\omega) .
\end{equation}
The limit of the expression $S_n$, as in (\ref{Rbar}), involves $\bar
h_k$, as in (\ref{eqLemBessel}). Using the formula (\ref{eqLemBessel})
for $\bar h_k(z^{-1})$ in terms of Bessel functions and Lemma \ref
{LemBoundsJandK},
one checks, introducing the running variable $a=\mu(2t)^{1/3}$,
%
\begin{eqnarray}\quad
\lim_{t\to\infty}\bar{h}_k(z^{-1})&=& -\lim_{t\to\infty}\sum
_{a\geq0}
(-z)^a J_{k+a+1}(4t)\nonumber\\
&=&-\lim_{t\to\infty}(2t)^{-1/3}\sum_{\kappa\geq\tilde\sigma
+(2t)^{-1/3}} e^{\mu(2t)^{1/3}\log(1-\zeta t^{-1/3})}{\cal
J}_{2t}^{(0)}(\kappa+\mu)\\
&=&-\int_0^\infty d\mu\, e^{-\mu\zeta2^{1/3}}
\Ai(\kappa+\mu).\nonumber
\end{eqnarray}
Therefore, one finds
%
\begin{eqnarray}\label{A5}
\lim_{t\to\infty}S_n(z^{-1})&=&\lim_{t\to\infty}\la Q,\chi_n \bar
h(z^{-1})\ra=\lim_{t\to\infty}\sum_{k\geq n}Q_k \bar
h_k(z^{-1})\nonumber\\
&=& \lim_{t\to\infty}(2t)^{-1/3}\sum_{\kappa\geq\tilde\sigma
+(2t)^{-1/3}} (2t)^{1/3} {\cal Q}_t(\kappa) \bar h_k(z^{-1}) \\
&=&- \int_{\kappa\geq\tilde\sigma} d\kappa\, {\cal Q}(\kappa)
\int
_0^\infty d\mu\, e^{-\mu\zeta2^{1/3}} \Ai(\kappa+\mu)=\hat{\cal
P}(\zeta).
\nonumber
\end{eqnarray}
This completes the proof of Lemma~\ref{Lemma71}.
\end{pf}

\textit{Sketch of Proof of Theorem}~\ref{ThmExtKernelAsympt},
\textit{formula} (\ref{ExtKernelB}).
Since the sum in brackets in (\ref{eq35ext}) is invariant under the
involution $x_1\leftrightarrow-x_2$ and $t_1\leftrightarrow-t_2$, it
suffices to consider the double integral, with the first term only. The
second half comes for free by acting with the involution!
Given scaling (\ref{sc}), Lemma~\ref{Lemma71} yields
%
\begin{equation}\label{A}
\lim_{t\to\infty} t^{1/3}\frac{dz \,dw}{z-w} \frac
{(-w)^{x_2-1}}{(-z)^{x_1}} \frac{e^{-t_1(z+z^{-1}+2)}}{e^{-t_2(w+w^{-1}+2)}}
=\frac{d\zeta \,d\omega}{\zeta-\omega}
\biggl(\frac{e^{\xi_1\zeta}}{e^{\xi_2\omega}}\biggr)
\frac{e^{s_1\zeta^2}}{e^{s_2\omega^2}}
\end{equation}
and, from (\ref{eqE1234}),
%
\begin{eqnarray}\label{B}
\lim_{t\to\infty}\sum_{i=1}^4 E_i(z,w)
&=& \frac{e^{{\zeta^3}/3-\sigma\zeta}}
{e^{{\omega^3}/3-\sigma\omega}} \bigl(1-\hat{\cal P}(\zeta)\bigr)
\bigl(1-\hat{\cal P}(-\omega)\bigr)\nonumber\\
&&{} - \frac{e^{{\zeta^3}/3-\sigma\zeta
}}{e^{-{\omega^3}/3+\sigma\omega}}
e^{2\sigma\omega} \bigl(1-\hat{\cal P}(\zeta)\bigr) \hat{\cal
Q}(-\omega)\nonumber\\[-8pt]\\[-8pt]
&&{} - \frac{e^{-{\zeta^3}/3+\sigma\zeta}}
{e^{{\omega^3}/3-\sigma\omega}}
e^{-2\sigma\zeta} \bigl(1-\hat{\cal P}(-\omega)\bigr) \hat{\cal
Q}(\zeta)\nonumber\\
&&{}- \frac{e^{{\zeta^3}/3-\sigma\zeta}}
{e^{{\omega^3}/3-\sigma\omega}}
\frac{e^{2\sigma\zeta} }{e^{2\sigma\omega}} \hat{\cal Q}(- \zeta
)\hat{\cal Q}(\omega).
\nonumber
\end{eqnarray}

Combining (\ref{A}) and (\ref{B}) yields the following limit below,
first with the contours as indicated in Figure~\ref{figPaths1}, which
then can be transformed into the vertical lines above in Figure \ref
{figPaths2}, compatible with Figure~\ref{figPaths1}.
%
\begin{figure}

\includegraphics{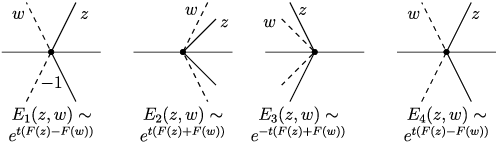}

\caption{Contours $z\in\Gamma_0$ and $w\in\Gamma_{0,z}$ in the
neighborhood of $z=w=-1$.}
\label{figPaths1}
\end{figure}
Indeed, to pick steepest descent paths about $z=w=-1$ respecting the
integration contours in $\oint_{\Gamma_0}dz \oint_{\Gamma_{0,z}} dw$
of (\ref{ExtKernelC}), one must choose the local paths, as illustrated
in Figure~\ref{figPaths1}; these paths must be completed by closed
%
\begin{figure}[b]

\includegraphics{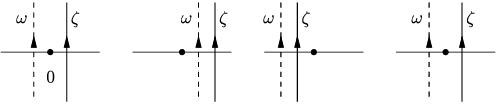}

\caption{Vertical lines $\pm\delta+\I\R$ and $\pm2 \delta+\I\R
$ of
integration for $\zeta$ and $\omega$.}
\label{figPaths2}
\end{figure}
contours encircling the origin deformed to provide steepest descent
contours. In the $\zeta,\omega$ scale, there are $4$ rays emanating
from the origin $\omega=\zeta=0$; one is then free to deform these rays
so as to obtain two parallel imaginary lines near the origin, as
depicted in Figure~\ref{figPaths2}. Therefore the following limit holds
for the first double integral
%
\begin{eqnarray}\label{ExtKernelC}
&&\lim_{t\to\infty}
\frac{t^{1/3}}{(2\pi\I)^2}\oint_{\Gamma_0}dz \oint_{\Gamma_{0,z}}
\frac{dw }{z-w}
\frac{(-w)^{x_2-1}}{(-z)^{x_1}}\frac
{e^{-t_1(z+z^{-1}+2)}}{e^{-t_2(w+w^{-1}+2)}}
\sum_{i=1}^4 E_i(z,w) \nonumber\\
&&\qquad=\frac{1}{(2\pi\I)^2}\int_{\delta+\I\R} d\zeta
\int
_{-\delta+\I\R} d\omega
\frac{e^{{\zeta^3}/3-\sigma\zeta}}{e^{{\omega
^3}/3-\sigma
\omega}}\frac{e^{s_1 \zeta^2}}{e^{s_2 \omega^2}}
\biggl(\frac{e^{\xi_1 \zeta}}{e^{\xi_2 \omega}}\biggr)\nonumber\\[-8pt]
&&\quad\hspace*{271pt}\mbox{\quad(i)}\nonumber\\[-8pt]
&&\qquad\quad\hspace*{118pt}{}\times
\frac{(1-\hat{\cal P}(\zeta))(1-\hat{\cal P}(-\omega))}{\zeta
-\omega
} \nonumber\\
&&\qquad\quad{} -\frac{1}{(2\pi\I)^2}\int_{2\delta+\I\R} d\zeta
\int
_{\delta+\I\R} d\omega
\frac{e^{{\zeta^3}/3-\sigma\zeta}}{e^{-{\omega
^3}/3-\sigma
\omega}}\frac{e^{s_1 \zeta^2}}{e^{s_2 \omega^2}}
\biggl(\frac{e^{\xi_1 \zeta}}{e^{\xi_2 \omega}} \biggr)\nonumber\\[-8pt]
&&\quad\hspace*{271pt}\mbox{\quad(ii)}\nonumber\\[-8pt]
&&\qquad\quad\hspace*{128pt}{}\times
\frac{(1-\hat{\cal P}(\zeta)) \hat{\cal Q}(-\omega)}{\zeta-\omega}
\\
&&\qquad\quad{} -\frac{1}{(2\pi\I)^2}\int_{-\delta+\I\R} d\zeta
\int
_{-2\delta+\I\R} d\omega
\frac{e^{-{\zeta^3}/3-\sigma\zeta}}{e^{{\omega
^3}/3-\sigma
\omega}} \frac{e^{s_1 \zeta^2}}{e^{s_2 \omega^2}}
\biggl(\frac{e^{\xi_1 \zeta}}{e^{\xi_2 \omega}}\biggr)\nonumber\\[-8pt]
&&\quad\hspace*{271pt}\mbox{\quad(iii)}\nonumber\\[-8pt]
&&\qquad\quad\hspace*{141pt}{}\times
\frac{(1-\hat{\cal P}(-\omega)) \hat{\cal Q}(\zeta)}{\zeta-\omega}
\nonumber\\
&&\qquad\quad{} -\frac{1 }{(2\pi\I)^2}\int_{\delta+\I\R} d\zeta
\int
_{-\delta+\I\R} d\omega
\frac{e^{{\zeta^3}/3+\sigma\zeta}}{e^{{\omega
^3}/3+\sigma
\omega}}
\frac{e^{s_1 \zeta^2}}{e^{s_2 \omega^2}}
\biggl(\frac{e^{\xi_1 \zeta}}{e^{\xi_2 \omega}}\biggr)\nonumber\\[-8pt]
&&\quad\hspace*{271pt}\mbox{\quad(iv)}\nonumber\\[-8pt]
&&\hspace*{164pt}{}\times
\frac{\hat{\cal Q}(-\zeta) \hat{\cal Q}(\omega)}{\zeta-\omega}.
\nonumber
\end{eqnarray}
In view of the scaling (\ref{sc}), the involution $x_1\leftrightarrow
-x_2$ and $t_1\leftrightarrow-t_2$ induces the involution $\xi
_1\leftrightarrow-\xi_2$ and $s_1\leftrightarrow-s_2$, so that the
limit of the other double integral is given by the same formula (\ref
{ExtKernelC}) above, but with
%
\begin{equation}
\xi_1\leftrightarrow-\xi_2 \quad\mbox{and}\quad
s_1\leftrightarrow-s_2.
\end{equation}

We are also allowed to interchange the integration variables $\zeta
\leftrightarrow-\omega$, provided the contours of integration are
modified accordingly; this last interchange implies
%
\begin{equation}
\int_{\delta+\I\R} d\zeta\int_{-\delta+\I\R
} d\omega\qquad\mbox{remains}
\end{equation}
%
\begin{equation}\quad
\int_{2 \delta+\I\R} d\zeta\int_{\delta+\I\R
} d\omega\quad\mbox{and}\quad
\int_{- \delta+\I\R} d\zeta\int_{-2\delta+\I\R
} d\omega\qquad\mbox{interchange}.
\end{equation}
So, the three combined maps,
%
\begin{equation}\label{invol}
\zeta\leftrightarrow-\omega,\qquad s_1 \leftrightarrow-s_2,\qquad \xi_1
\leftrightarrow-\xi_2
\end{equation}
have the following effect on the four double integrals $\mbox{(i)},\ldots,
\mbox{(iv)}$ in (\ref{ExtKernelC}):
\begin{eqnarray*}
\mbox{double integral (i) with $\dfrac{e^{\xi_1 \zeta}}{e^{\xi_2
\omega
}}$}&\to&\mbox{same double integral (i), except for $\dfrac{e^{-\xi_1
\zeta}}{e^{-\xi_2 \omega}}$};\\
\mbox{double integral (ii) with $\dfrac{e^{\xi_1 \zeta}}{e^{\xi_2
\omega
}}$}&\to&\mbox{same double integral (iii), except for $\dfrac{e^{-\xi_1
\zeta}}{e^{-\xi_2 \omega}}$};\\
\mbox{double integral (iii) with $\dfrac{e^{\xi_1 \zeta}}{e^{\xi_2
\omega}}$}&\to&\mbox{same double integral (ii), except for $\dfrac
{e^{-\xi_1 \zeta}}{e^{-\xi_2 \omega}}$};\\
\mbox{double integral (iv) with $\dfrac{e^{\xi_1 \zeta}}{e^{\xi_2
\omega
}}$}&\to&\mbox{same double integral (iv), except for $\dfrac{e^{-\xi_1
\zeta}}{e^{-\xi_2 \omega}}$.}
\end{eqnarray*}
Therefore the limit
%
\begin{eqnarray}
&&\lim_{t\to\infty} \frac{t^{1/3}}{(2\pi\I)^2}\oint_{\Gamma_0}dz
\oint
_{\Gamma_{0,z}}\frac{dw}{z-w}\sum_{i=1}^4 E_i(z,w)
\nonumber\\[-8pt]\\[-8pt]
&&\qquad{}\times\biggl({\frac{(-w)^{x_2-1}}{(-z)^{x_1}}\frac
{e^{-t_1(z+z^{-1}+2)}}{e^{-t_2(w+w^{-1}+2)}}} + {\frac
{(-z)^{x_2}}{(-w)^{x_1+1}}\frac
{e^{-t_1(w+w^{-1}+2)}}{e^{-t_2(z+z^{-1}+2)}}} \biggr)
\nonumber
\end{eqnarray}
is given by the right-hand side of (\ref{ExtKernelC}) with the replacement
%
\begin{equation}
\frac{e^{\xi_1 \zeta}}{e^{\xi_2 \omega}} \rightarrow\frac
{e^{\xi_1
\zeta}}{e^{\xi_2 \omega}} +\frac{e^{-\xi_1 \zeta}}{e^{-\xi_2
\omega}}.
\end{equation}
Finally, in order to change the sign of the last integral, one switches
the sign $\omega\rightarrow-\omega$ and $\zeta\rightarrow-\zeta$,
which changes
%
\begin{equation}\qquad
-\int_{\delta+\I\R} d\zeta
\int_{-\delta+\I\R} d\omega\frac{1}{\zeta-\omega
}
\qquad\mbox{into }
+\int_{- \delta+\I\R} d\zeta
\int_{\delta+\I\R} d\omega
\frac{1}{\zeta-\omega}.
\end{equation}
Renaming variables $\zeta\to u, \omega\to v$ gives formula (\ref
{ExtKernelB}).

\begin{appendix}\label{app}
\section*{Appendix: Some properties of Bessel and Airy functions}
Let us recall that the Bessel function representation of order $n\in\Z$
%
\setcounter{equation}{0}
\begin{equation}
J_n(2t)=\frac1{2\pi\I}\oint_{\Gamma_0}dz\frac{e^{t(z-z^{-1})}}{z^{n+1}}
\end{equation}
has the symmetries
%
\begin{equation}
J_n(2t)=(-1)^n J_{-n}(2t)=(-1)^n J_n(-2t).
\end{equation}
Moreover,
%
\begin{equation}\label{eqA0}
\frac{1}{2\pi\I}\oint_{\Gamma_0}\frac{dz}{z}\frac
{e^{b(z-z^{-1})}e^{a(z+z^{-1})}}{z^n} =\biggl(\frac{b+a}{b-a}
\biggr)^{n/2} J_n\bigl(2\sqrt{b^2-a^2}\bigr).
\end{equation}
It is well known~\cite{AS84} that
%
\begin{equation}\label{eqA1}
\lim_{t\to\infty} t^{1/3} J_{[2t+\xi t^{1/3}]}(2t)=\Ai(\xi).
\end{equation}
An uniform bound obtained in~\cite{Lan00} is
%
\begin{equation}
|(2t)^{1/3} J_{n}(2t)|\leq c,\qquad c=0.785\ldots,\qquad n\in\Z.
\end{equation}
This bound, together with uniform expansion which can be found in \cite
{AS84} is used in Lemma A.1 of~\cite{Fer04} to get the following
result. Fix any $\theta>0$. Then, there exists a constant $t_0>0$ and a
constant $C>0$ such that, uniformly in $t\geq t_0$,
%
\begin{equation}\label{eqA3}
\bigl|t^{1/3} J_{[2t+\xi t^{1/3}]}(2t)\bigr|\leq C \min\{1,e^{-\theta\xi}\}.
\end{equation}
Actually, the statement of Lemma A.1 of~\cite{Fer04} is for $\theta
=1/2$ but inspecting the proof it is straightforward to see that it
holds for any fixed $\theta>0$.
The Airy function has, among others, the following two integral
representations. For any $\delta>0$, it holds
%
\begin{equation}\label{A4k}\hspace*{23pt}
\Ai(x)=\frac{1}{2\pi\I}\int_{\delta+\I\R}du\,
e^{u^3/3-ux},\qquad \Ai
(x)=\frac{1}{2\pi\I}\int_{-\delta+\I\R}dv\,e^{-v^3/3+vx}.
\end{equation}
Moreover, for any $\delta>0$, it holds
%
\begin{eqnarray}\label{eqA10}
\mathrm{Ai}^{(s)}(x)=e^{s x +2s^3/3}\Ai(x+s^2) &=& \frac{1}{2\pi\I
}\int
_{\delta+\I\R}du\, e^{u^3/3+u^2 s-ux},\nonumber\\[-8pt]\\[-8pt]
\mathrm{Ai}^{(s)}(x)=e^{s x +2s^3/3}\Ai(x+s^2) &=& \frac{1}{2\pi\I
}\int
_{-\delta+\I\R}dv\, e^{-v^3/3+v^2s+vx}.
\nonumber
\end{eqnarray}
\end{appendix}



\printaddresses

\end{document}